\documentclass[lettersize,journal]{IEEEtran}
\usepackage[utf8]{inputenc}
\usepackage{amsmath,amsfonts}

\usepackage{algorithmic}
\usepackage{algorithm}
\usepackage{array}
\usepackage[caption=false,font=normalsize,labelfont=sf,textfont=sf]{subfig}
\usepackage{textcomp}
\usepackage{stfloats}
\usepackage{url}
\usepackage{verbatim}
\usepackage{graphicx}
\usepackage{xcolor}
\usepackage{cite}
\usepackage{lipsum,extarrows}
\usepackage{cuted}
\usepackage{bm}
\usepackage{cleveref}
\usepackage{multirow}
\usepackage{amsthm, amsfonts, amsmath, amssymb, mathrsfs, extarrows}
\usepackage{diagbox}
\newtheorem{proposition}{Proposition}
\usepackage[margin=0pt,font=small]{caption}

\begin{document}
\title{Joint Beam Scheduling and Beamforming Design for Cooperative Positioning in Multi-beam LEO Satellite Networks}

 \author{{Hongtao~Xv,
             Yaohua~Sun,
             Yafei~Zhao,
             Mugen~Peng,~\IEEEmembership{Fellow,~IEEE},
             and Shijie Zhang}

   \thanks{This work is supported in part by the National Natural Science Foundation of China (Grant No. 62371071 and No. 62001053), and in part by the Young Elite Scientists Sponsorship Program by CAST (Grant No. 2021QNRC001).}
       \thanks{Copyright (c) 2015 IEEE. Personal use of this material is permitted. However, permission to use this material for any other purposes must be obtained from the IEEE by sending a request to pubs-permissions@ieee.org. }
     \thanks{Hongtao~Xv (htxu@bupt.edu.cn), Yaohua~Sun (sunyaohua@bupt.edu.cn), Yafei~Zhao(zhaoyafei@bupt.edu.cn), and Mugen~Peng (pmg@bupt.edu.cn) are with the State Key Laboratory of Networking and Switching Technology, Beijing University of Posts and Telecommunications, Beijing 100876, China. 
     Shijie Zhang (zsj@yinhe.ht) is with Beijing University of Posts and Telecommunications and Yinhe Hangtian (Beijing) Internet Technology Co., Ltd..
     (\bf{Corresponding author: Yaohua Sun})}
 }
\maketitle
\begin{abstract}
Cooperative positioning with multiple low earth orbit (LEO) satellites is promising in providing location-based services and enhancing satellite-terrestrial communication.
However, positioning accuracy is greatly affected by inter-beam interference and satellite-terrestrial topology geometry. To select the best combination of satellites from visible ones and suppress inter-beam interference, this paper explores the utilization of flexible beam scheduling and beamforming of multi-beam LEO satellites that can adjust beam directions toward the same earth-fixed cell to send positioning signals simultaneously.
By leveraging Cram\'{e}r-Rao lower bound (CRLB) to characterize user Time Difference of Arrival (TDOA) positioning accuracy, the concerned problem is formulated, aiming at optimizing user positioning accuracy under beam scheduling and beam transmission power constraints.
To deal with the mixed-integer-nonconvex problem, we decompose it into an inner beamforming design problem and an outer beam scheduling problem. For the former, we first prove the monotonic relationship between user positioning accuracy and its perceived signal-to-interference-plus-noise ratio (SINR) to reformulate the problem, and then semidefinite relaxation (SDR) is adopted for beamforming design.
For the outer problem, a heuristic low-complexity beam scheduling scheme is proposed, whose core idea is to schedule users with lower channel correlation to mitigate inter-beam interference while seeking a proper satellite-terrestrial topology geometry.
Simulation results verify the superior positioning performance of our proposed positioning-oriented beamforming and beam scheduling scheme, and it is shown that average user positioning accuracy is improved
by $17.1\%$ and $55.9\%$ when the beam transmission power is 20 dBw, compared to conventional beamforming and beam scheduling schemes, respectively.
\end{abstract}


\section{Introduction}
With the evolution of small satellite platform manufacture and space launch technology, dense LEO satellite constellations will come into being, such as Starlink and OneWeb \cite{1,res7}.
In such constellations, multiple LEO satellites can simultaneously cover a ground user, which facilitates providing more enhanced services, including satellite edge computing \cite{0908-1} and high precision user positioning \cite{0908-2}. For the latter, multiple satellites can send positioning signals to the same user or cell to improve Geometric Dilution of Precision (GDOP), convergence speed, and pseudo-range measurement accuracy, where cooperative satellite beam scheduling and beamforming design are the fundamental issues.

Owing to the development of software-defined payloads, satellite onboard communication resources can be scheduled and configured on demand. An LEO satellite can flexibly adjust the serving relationship between beams and earth-fixed cells \cite{aa2}, following a specific time-space transmission plan. Such flexibility in beam scheduling enables a good match between limited network resources and heterogeneous service demands, thereby greatly enhancing system utility \cite{aa4,res8}. The design of beam scheduling usually involves binary variables, which makes the problem fall into the category of integer programming that is difficult to solve. The authors of \cite{aa1} made use of the Semi-Definite Programming (SDP) technique and presented a heuristic algorithm to provide a solution rapidly. Aiming to achieve global optimal beam scheduling, the authors of \cite{aa7} employed a simulated annealing method at the expense of high computation complexity. Considering the unknown dynamics of the satellite communication environment, a model-free deep reinforcement learning-based approach was proposed to learn the optimal beam scheduling plan \cite{aa3}. However, these previous works ignore the potential co-frequency interference among beams. To further reduce the inter-beam interference, literature \cite{aa5} provided a novel joint beam scheduling and SVD-based beamforming approach to suppress inter-beam interference. The authors in \cite{res1,res2,res3} adopted SDR to reformulate the co-channel interference suppression problem and conducted successive convex approximation (SCA) to optimize the beamformer. In a similar vein, the authors in \cite{0908-1} adopted Lyapunov optimization theory to solve the problems of binary association and beamforming to maximize the data rate. Moreover, literature \cite{aa1} jointly considered beam scheduling design and Zero Forcing (ZF) beamforming under per feed power constraint at the satellite.


The above works on satellite beam scheduling and beamforming mainly focus on satisfying user communication demands. However, multi-beam LEO satellite systems need to take advantage of positioning and provide positioning information services in the future \cite{2}. Compared with GNSS, which relies on satellites mainly operating in medium earth orbits, the free space loss with LEO satellites is much lower. The received power on the ground can be improved by more than 30 dB under the same transmission power \cite{3}, which makes it easier for us to capture and track positioning signals. Apart from providing location-based services, position information is also essential for enhancing communication services. In the physical layer, positioning information can be used for improving carrier frequency offset estimation in uplink orthogonal frequency division multiple access systems \cite{pos_for_fre}. In the networking layer, a user terminal (UT) can leverage its position information to trigger a handover procedure between satellites actively \cite{pos_for_handover}.

A well-known methodology for positioning is the time difference of arrival (TDOA) technique \cite{a5}. In this paradigm, a UT calculates the difference between the times at which the signals of a pair of satellites arrive at the UT. For each pair of satellites, TDOA provides a hyperbolic map plot of possible UT positions \cite{6}. Each additional satellite pair provides another intersecting hyperbola, further refining the UT position.
To assess TDOA performance \cite{7,a4,a9,a10,a11}, the metric for benchmarking any unbiased estimators, CRLB, is generally used.
There has been some research on optimizing TDOA positioning performance by beamforming and beam scheduling in terrestrial and satellite networks.
The authors of \cite{a6} optimized beamforming vectors at base stations (BSs) to minimize the power expenditure under data rate and positioning accuracy constraints, where multiple BSs with multiple antennas communicate with multiple single-antenna UTs. Literature \cite{a7} proposed a method to decrease the overlap of estimated positions as well as increase the accuracy of TDOA positioning by using beamforming at BSs. In LEO satellite systems, the authors of \cite{10} investigated the beam scheduling problem and designed a novel beam hopping framework to optimize positioning performance.

Overall, the joint design of beamforming and beam scheduling for positioning in multi-beam LEO satellite systems still needs to be explored, which is technically challenging for several reasons. First, with the number of visible satellites increasing, it becomes intractable to quickly select serving satellites for UTs/cells to maximize positioning accuracy. Second, considering full frequency reuse (FFR) among satellite beams \cite{8,9}, inter-beam interference becomes critical, and the radio ranging measurement accuracy can significantly degrade as SINR gets lower. Third, mitigating inter-beam interference also imposes higher requirements on beam scheduling design, and meanwhile, proper satellite-terrestrial topology geometry should also be formed by beam schedule. Our main contributions are summarized as follows.

\begin{itemize}
    \item  A novel multi-beam LEO satellite network scenario for cooperative UT positioning is considered, where each satellite can generate multiple beams using a planar antenna array, and each UT can receive positioning signals sent by beams from numerous satellites.
    Using CRLB to characterize the TDOA positioning accuracy of a single UT, we formulate the problem of joint satellite beam scheduling and beamforming design, aiming at optimizing UT positioning accuracy under per-beam transmission power constraint and the constraint on the number of serving beams for each UT.

    \item To deal with the non-convex mixed-integer problem, it is decomposed into an inner beamforming design problem and an outer beam scheduling problem. For the former, we propose a UT SINR threshold adjustment based beamforming algorithm with the SDR technique, which is inspired by the monotonic relationship between the positioning accuracy of a single UT and its perceived SINR. For the latter, a fast and efficient greedy heuristic beam scheduling algorithm that changes the value of a 0-1 variable from 0 to 1 in each round of iteration is developed, considering the trade-off between channel correlation and UT-satellite topology geometry.

    \item Extensive numerical evaluations are carried out, where our proposal is compared with conventional beamforming and GDOP-based beam scheduling schemes under different beam transmission powers and the numbers of visible satellites. It is shown that average user positioning accuracy can be improved
    by $17.1\%$ and $55.9\%$, respectively.
\end{itemize}

The remainder of the paper is organized as follows. Section \ref{sec2} presents the system model and the concerned positioning accuracy optimization problem. Section \ref{section3} details the design of the beamforming algorithm. In Section \ref{section4}, the beam scheduling problem is solved. Simulation results are given in Section \ref{sec_SIMULATION_RESULTS}, and the conclusion is drawn in Section \ref{section5}. Essential symbols are summarized in Table \ref{table4}.

We use the following notation throughout this paper:
$\mathbf{X}$ is a matrix; $\mathbf{x}$ is a vector; $x$ is a scalar; $ \mathbf{X}(:,j)$ is the  $j$-th column of $\mathbf{X}$; $\mathbf{X}(i,j)$ is the element in the $i$-th row and $j$-th column of $\mathbf{X}$; The transpose and conjugate
transpose of $\mathbf{X}$ are represented by $\mathbf{X}^T$ and $\mathbf{X}^H$, respectively; $\mathbf{I}_N$ is the $N \times
N$ identity matrix; $\text{rank}(\mathbf{X})$ denotes
the \text{rank} of $\mathbf{X}$; $\mathbf{X}^{-1}$ is the inverse of $\mathbf{X}$; $\lambda_i(\mathbf{X})$ denotes the $i$-th eigenvalue of $\mathbf{X}$; $\text{tr}(\mathbf{X})$ indicates the trace of $\mathbf{X}$; $\vert \vert \mathbf{x}\vert \vert_2$ is 2-norm of $\mathbf{x}$; $\mathbb{C}^{m\times n}$ denotes an $m$ by $n$ dimensional complex space; Kronecker product is denoted by $\otimes$; $\mathbb{E}$ is the expectation operator;
$\frac{\partial f}{\partial x}$ is used to denote the partial derivative of function $f$ with respect to $x$.

\begin{table}[!t]
	\centering \caption{Symbol summary}
	\label{table4}
	\begin{tabular}{l l}  \hline
		\textbf{Notation} & \textbf{Definition} \\
		\hline
		$\mathcal{I}$ & The set of visible satellites \\
		$\mathcal{C}$ & The set of UTs \\
		$\mathcal{I}_c$ & The set of satellites serving UT $c$ \\
		$\mathcal{C}_i$ & The set of UTs served by satellite $i$ \\
		$N$ & The number of antenna elements at each satellite\\
		$\delta_{i,c}$ & A 0-1 indicator representing whether satellite $i$ serves UT $c$\\
		$\sigma_{i,c}$ & The CRLB of MSE of TOA measurement of UT $c$ from\\& satellite $i$\\
		$\mathbf{w}_{i,c}$ & The beamforming vector of satellite $i$ for UT $c$\\
		$\text{SINR}_{i,c}$ & The SINR of UT $c$ from satellite $i$\\
		$\mathbf{h}_{i,c}$ & The channel coefficients between satellite $i$ and UT $c$\\
		$\mathbf{s}$ & The three-dimensional coordinate of satellite or UT\\
		$d_{i,c}$ & The distance between satellite $i$ and UT $c$\\
		$\gamma_{i,c}$ & The SINR threshold of UT $c$ from satellite $i$\\
		$\text{CRLB}_{c}$ & The positioning accuracy of UT $c$ in terms of CRLB\\
		$K$ & The maximum number of beams of one satellite \\
		$I_{\text{TDOA}}$ & The number of satellites serving each UT except the reference \\&satellite \\
		$P$ & The maximum beam transmission power \\
		$B$ & The bandwidth of each satellite\\
		\hline
	\end{tabular}
\end{table}

\section{System Model and Problem Formulation}
\label{sec2}
\subsection{System Model}

The LEO satellite communication system scenario for cooperative UT positioning is shown in Fig.\ref{fig1}. 
To realize centralized control and coordinate satellites, if the reference satellite has powerful computing capability and connects to other satellites through inter-satellite links (ISLs), the positioning-oriented beam scheduling and beamforming algorithms are executed at this satellite, whose results are fed back to other satellites via ISLs. Otherwise, the gateway (GW), covering the cooperative LEO satellites simultaneously, is responsible for executing these algorithms and forwards the results to satellites via feeder links.
Each satellite is equipped with a uniform planar array (UPA) composed of $N=N_xN_y$ antennas, where $N_x$ and $N_y$ are the numbers of antennas on the x and y axis, respectively, and it is assumed that each satellite can generate at most $K$ beams.
By sending positioning signals on time-frequency resource that is orthogonal to the resource used for normal communication services,
multiple satellites can cooperatively provide UT positioning capability.
There are multiple single-antenna UTs distributed in $C$ earth-fixed cells whose set is denoted by $\mathcal{C}=\{1,2,...,C\}$, and denote the set of LEO satellites visible to the cells by $\mathcal{I}=\{1,2,...,I\}$. Further, it is assumed that there is one UT in each cell for simplicity, and hence we use UT $c$ to represent the UT located in cell $c$ in the following.
Define $\delta_{i,c}=1$ if satellite $i$ schedules a beam to serve UT $c$, and $\delta_{i,c}=0$ otherwise. The set of satellites serving UT $c$ and the set of UTs served by satellite $i$ are denoted as $\mathcal{I}_c=\{i|\delta_{i,c}=1\}$ and $\mathcal{C}_i=\{c|\delta_{i,c}=1\}$, respectively.

The beamforming vector of satellite $i$ for UT $c$ is denoted as $\mathbf{w}_{i,c}\in \mathbb{C}^{N\times 1}$ and the collection of all beamforming vectors is denoted as $\mathbf{W}_i\in \mathbb{C}^{N\times \vert\mathcal{C}_i\vert}$. Then, the discrete-time transmitted navigation signal of satellite $i$ is therefore given by $\mathbf{x}_{i}=\mathbf{W}_{i}\mathbf{n}_{i}$, where $\mathbf{n}_i$ is the $\vert\mathcal{C}_i\vert \times 1$ symbol vector 
consisting of the navigation signal $n_{i,c}$ for UT $c$ with $\mathbb{E}[\mathbf{n}_i\mathbf{n}_i^H] = \mathbf{I}_{\vert\mathcal{C}_i\vert}$. Suppose that proper frequency reuse is utilized among satellites so there is no inter-satellite interference. In contrast, intra-satellite interference exists due to full frequency reuse among beams of the same satellite. Therefore, the received signal of UT $c$ from satellite $i$ is given by

\begin{equation}
   y_{i,c}=\mathbf{h}_{i,c}^T\mathbf{w}_{i,c}n_{i,c}+\sum_{c'\in \mathcal{C}_i,c'\neq c}\mathbf{h}_{i,c}^T\mathbf{w}_{i,c'}n_{i,c'}+\hat{y}_{i,c},
\end{equation}
where $\mathbf{h}_{i,c}\in \mathbb{C}^{N\times 1}$ is the channel vector between satellite $i$ and UT $c$, and $\hat{y}_{i,c} \sim \mathcal{N}(0,\epsilon_c^2)$ is white Gaussian noise. Then, the SINR of UT $c$ when receiving the positioning signal from satellite $i$ is given by
\begin{equation}\label{2}
	\text{SINR}_{i,c}=\frac{\lvert\mathbf{h}_{i,c}^T\mathbf{w}_{i,c}\rvert^2}{\sum_{c'\in\mathcal{C}_i, c'\neq c}\lvert\mathbf{h}_{i,c}^T\mathbf{w}_{i,c'}\rvert^2+\epsilon_c^2},\quad i\in \mathcal{I}_c.
\end{equation}

Considering the clock asynchronous problem, the signal from a reference satellite is used as the reference measurement for all UTs. With the knowledge of satellite positions acquired from the ephemeris, UTs use TDOA measurements for positioning. Specifically, UT $c$ uses the signal from satellite $i\in\mathcal{I}_c$ to estimate time of arrival (TOA) and calculates the time difference between TOAs of satellite $i$ and the reference satellite. Then, 
UT $c$ estimates its own position with the observation equations of TDOA. According to \cite{7,12}, for UT $c$ and satellite $i$, the CRLB of TOA measurement under an AWGN channel is calculated by
\begin{equation}\label{3}
	\sigma_{i,c}^2= \frac{3}{4\pi^2B^2 \text{SINR}_{i,c}},
\end{equation}
where $B$ is the bandwidth of satellite $i$'s beam \cite{13}.
\begin{figure}[htbp]
    \color{red}
	\centering
	\includegraphics[width=0.5\textwidth]{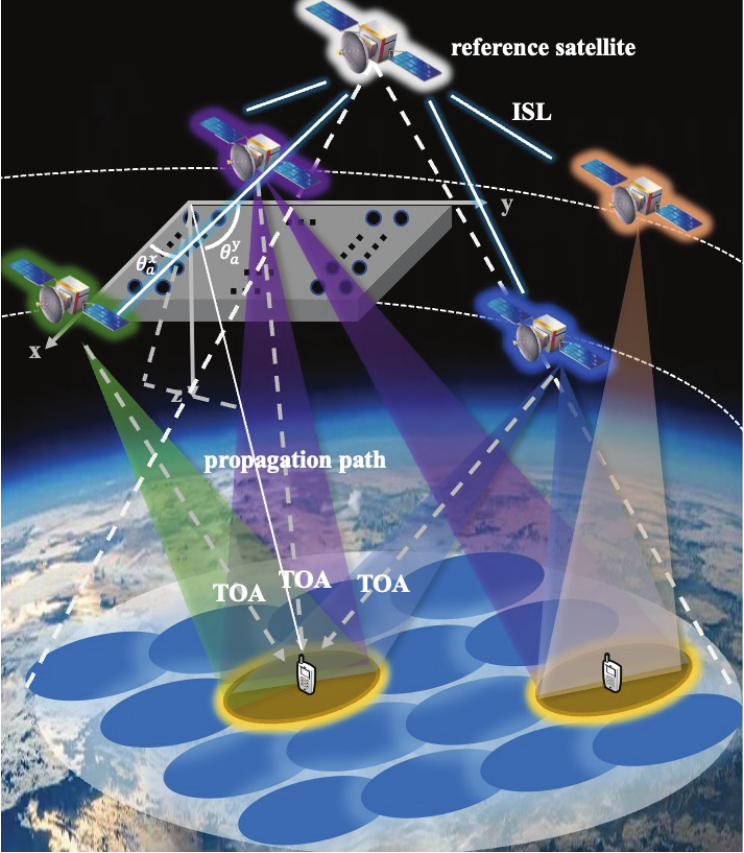}
	\caption{The multi-beam LEO satellite system scenario for cooperative UT positioning.}
	\label{fig1}
\end{figure}

Denote the positions of UT $c$, satellites $i$ and the reference satellite as $\mathbf{s}_c=\left[x_c,y_c,z_c\right]^T$, $\mathbf{s}_i=\left[x_i,y_i,z_i\right]^T$ and $\mathbf{s}_0=\left[x_0,y_0,z_0\right]^T$, respectively.
Then, the distance $d_{i,c}$ between satellite $i$ and UT $c$ can be expressed as
\begin{equation}\label{1}
	d_{i,c}=\Vert \mathbf{s}_c-\mathbf{s}_i \Vert_2.
\end{equation}
Thus, UT $c$ can get TDOA measurements for satellite $i$ and the reference satellite as
\begin{equation}\label{TDOA}
    \begin{split}
        \text{TDOA}_{i,c}&=\frac{1}{v}(d_{i,c}-d_{0,c})+(\text{TOA}_{i,c}-\text{TOA}_{0,c})\\
        &=G_{i,c}(\mathbf{s}_c)+(\text{TOA}_{i,c}-\text{TOA}_{0,c}),
    \end{split}
\end{equation}
where $v$ is the speed of light, $G_{i,c}(\mathbf{s}_c)=\frac{1}{v}(d_{i,c}(\mathbf{s}_c)-d_{0,c}(\mathbf{s}_c))$, $\text{TOA}_{i,c}$ and $\text{TOA}_{0,c}$ are the TOA measurement and $\text{TOA}_{i,c}\sim \mathcal{N}(0,\sigma_{i,c}^2)$. As a result of the fact of a common reference satellite, TDOA measurements are correlated. Denote the indexes of satellites serving UT $c$ as $i_1,i_2,\cdots,i_{\vert \mathcal{I}_c\vert}$, the variance of $\text{TOA}_{i_1,c}-\text{TOA}_{0,c}$ is $\sigma_{i_1,c}^2+\sigma_{0,c}^2$ and $\text{TDOA}_c=[\text{TDOA}_{i_1,c},\text{TDOA}_{i_2,c},\cdots,\text{TDOA}_{i_{\vert \mathcal{I}_c\vert},c}]^T$ has a joint conditional Gaussian distribution with covariance matrix $\mathbf{R}_c$ given as
\begin{equation*}
	\mathbf{R}_c=\left[
	\begin{array}{cccc}
		\sigma_{0,c}^2+\sigma_{{i_1},c}^2 & \sigma_{0,c}^2 & \cdots &\sigma_{0,c}^2 \\
		\sigma_{0,c}^2 & \sigma_{0,c}^2+\sigma_{{i_2},c}^2  & \cdots &\sigma_{0,c}^2 \\
		\vdots & \vdots & \ddots&\vdots \\
		\sigma_{0,c}^2 & \sigma_{0,c}^2 & \cdots&\sigma_{0,c}^2 +\sigma_{i_{{\left|\mathcal{I}_c\right|}},c}^2 \\
	\end{array}
	\right].
\end{equation*}
The measurement equation for the measurement vector $\text{TDOA}_c$ is
given by
\begin{equation}
    \text{TDOA}_c=G_c(\mathbf{s}_c)+u_c, u_c\sim \mathcal{N}(0,\mathbf{R}_c),
\end{equation}
where $G_c(\mathbf{s}_c)=[G_{i_1,c}(\mathbf{s}_c),G_{i_2,c}(\mathbf{s}_c),\cdots,G_{i_{\vert \mathcal{I}_c\vert,c}(\mathbf{s}_c),c}]^T$.

The CRLB for TDOA positioning equals the inverse of the Fisher information $\mathbf{J}$ \cite{res5} as
\begin{align}\label{4}
    \text{CRLB}_c=tr(\mathbf{J}^{-1})&=tr(((\frac{\partial G_c(\mathbf{s}_c)}{\partial \mathbf{s}_c})^T\mathbf{R}_c^{-1}(\frac{\partial G_c(\mathbf{s}_c)}{\partial \mathbf{s}_c}))^{-1})\notag \\ 
    &=tr((\mathbf{A}_c^T\mathbf{R}_c^{-1}\mathbf{A}_c)^{-1}).
\end{align}
$\mathbf{A}_c$ represents the Jacobian matrix of $G_c$ with regard to $\mathbf{s}_c$, which is given by
\begin{equation*}
    	\mathbf{A}_c=\frac{\partial G_c(\mathbf{s}_c)}{\partial \mathbf{s}_c}=\frac{1}{v} \left[\mathbf{a}_{i_1,c},\cdots,\mathbf{a}_{i_{\left|\mathcal{I}_c\right|},c} \right]^T,
\end{equation*}
$\mathbf{a}_{i,c}=\frac{\mathbf{s}_c-\mathbf{s}_i}{d_{i,c}}-\mathbf{a}_{0,c}$, and $\mathbf{a}_{0,c}=\frac{\mathbf{s}_c-\mathbf{s}_0}{d_{0,c}}$. The position error in meters for UT $c$ is finally computed as $\sqrt{\text{CRLB}_c}$ \cite{12}. 

\subsection{Channel Model}
Particularly, the free space path loss is considered to construct the channel model. According to \cite{res1,res2,res3}, the channel vector from satellite $i$ to UT $c$ is modeled as
\begin{equation}
\mathbf{h}_{i,c}=\sqrt{L_{i,c}}e^{j\theta_{i,c}}\mathbf{v}_{i,c},
\end{equation}
where $\theta_{i,c}$ is the phase vector with uniform distribution over $[0,2\pi)$,
$L_{i,c}$ is the free space loss given by
\begin{equation}
L_{i,c}(\text{dB})=20lg(f)+20lg(d_{i,c})+32.4,
\end{equation}
and $\mathbf{v}_{i,c}$ is the UPA response vector.
Based on \cite{14}, the UPA response vector $\mathbf{v}_{i,c}$ can be further calculated as
\begin{equation}\label{7}
 \mathbf{v}_{i,c}=\mathbf{v}_x\left(\theta_{i,c}^x\right)\otimes\mathbf{v}_y\left(\theta_{i,c}^y\right),
\end{equation}
\begin{equation*}
	\mathbf{v}_x\left(\theta_{i,c}^x\right)=\frac{1}{N_x}\left[1 , e^{-j\pi\theta_{i,c}^x}, ..., e^{-j\pi(N_x-1)\theta_{i,c}^x}\right]^T,
\end{equation*}
\begin{equation*}
	\mathbf{v}_y\left(\theta_{i,c}^y\right)=\frac{1}{N_y}\left[1 , e^{-j\pi\theta_{i,c}^y}, ..., e^{-j\pi(N_y-1)\theta_{i,c}^y}\right]^T,
\end{equation*}
where $\mathbf{v}_x$ and $\mathbf{v}_y$ are the array response vector of the angle with respect to the x and y-axis.
The parameter $\theta_{i,c}^x$ and $\theta_{i,c}^y$ are related to physical angles via $\theta_{i,c}^x=\sin(\varphi_{i,c}^y)\cos(\varphi_{i,c}^x)$ and $\theta_{i,c}^y=\cos(\varphi_{i,c}^y)$, where $\varphi_{i,c}^x$ and $\varphi_{i,c}^y$ are the angles with respect to the x and y axis associated with the propagation path from $i$-th satellite to UT $c$.

\subsection{Problem Formulation}
The positioning-oriented beam scheduling and beamforming design is formulated as the following problem, aiming to optimize UT TDOA positioning accuracy, which is given by
\begin{equation}\label{5}
    \begin{split}
	&\mathcal{P}_0:\min\limits_{\{\mathbf{w}_{i,c}\}, \{\delta_{i,c}\}} \sum\limits_{c=1}^{C} F_c\left(\mathbf{R}_c\left(\{\mathbf{w}_{i,c}\}, \{\delta_{i,c}\}\right),\mathbf{A}_c\left(\{\delta_{i,c}\}\right)\right)  \\%
	s.t.\quad& C_1:\ \sum\limits_{c=1}^C\delta_{i,c}\leq K, \ \forall i  \in \mathcal{I},                 \\
	& C_2:\ \sum\limits_{i\in \mathcal{I}}\delta_{i,c}= I_{\text{TDOA}},\ \forall c \in \mathcal{C},       \\	
	& C_3:\ \delta_{i,c}\in \{0,1\},\ \forall i  \in \mathcal{I} ,\ \forall c \in \mathcal{C},       \\	
	& C_4:\ \vert\vert \mathbf{w}_{i,c}\vert\vert_2^2=P\cdot \delta_{i,c},\ \forall i  \in \mathcal{I} ,\ \forall c \in \mathcal{C},
    \end{split}
\end{equation}
where {$F_c=\text{CRLB}_c$}. Constraint $C_1$ states that each satellite can generate at most $K$ beams simultaneously due to limited hardware capability. Constraint $C_2$ ensures that each UT receives positioning signals from $I_{\text{TDOA}}$ beams from different satellites alongside the reference satellite. Constraint $C_3$ means $\delta_{i,c}$ is a 0-1 indicator variable. Constraint $C_4$ is the beam transmission power constraint, and $\mathbf{w}_{i,c}$ is an all-zero column vector of dimension $N$ when $\delta_{i,c}=0$.

\section{Dynamic UT SINR Threshold Adjustment Based Satellite Beamforming}
\label{section3}

In \eqref{5}, we seek to find beamforming design $\{\mathbf{w}_{i,c}\}$ and beam scheduling plan $\{\delta_{i,c}\}$ that optimize user positioning accuracy. Unfortunately, finding a globally optimal solution in polynomial time is impossible due to the non-convex objective function and mixed-integer structure.
To make the problem more tractable, in this section, we first deal with positioning-oriented beamforming under any pre-fixed $\{\delta_{i,c}\}$,
and the corresponding problem is given as
\begin{equation}\label{8}
\begin{split}
 	&\mathcal{P}_1:\min\limits_{\{\mathbf{w}_{i,c}\}} \sum\limits_{c=1}^C F_c(\{\mathbf{w}_{i,c}\}\vert \{\delta_{i,c}\}) \\
	s.t. \quad& C_4:\ \vert\vert \mathbf{w}_{i,c}\vert\vert_2^2=P\cdot \delta_{i,c},\ \forall i  \in \mathcal{I} ,\ \forall c \in \mathcal{C}.
\end{split}
\end{equation}

In the following proposition, we first highlight the intrinsic relationship between UT's positioning accuracy and UT's perceived SINR from a single satellite to guide beamforming algorithm design.
\begin{proposition}
  $F_c$ is monotonically decreasing with respect to $\text{SINR}_{i,c}$ with the gradient of $F_c$ with respect to $\text{SINR}_{i,c}$ given by
\begin{equation}
 \frac{\partial F_c}{\partial \text{SINR}_{i,c}}=-\frac{3\text{tr}({\mathbf{R}_c^{-1}(\widehat{i},:)\mathbf{A}_c(\mathbf{A}_c^T\mathbf{R}_c^{-1}\mathbf{A}_c)^{-2}\mathbf{A}_c^T \mathbf{R}_c^{-1}(:,\widehat{i})})}{4\pi^2B^2\text{SINR}_{i,c}^2},
\end{equation}
where satellite $i$ is the $\widehat{i}$-th satellite in set $\mathcal{I}_c$ serving UT $c$, $\mathbf{A_c}$ is a column full rank matrix with probability one and $\mathbf{R_c}$ is positive definite.
\end{proposition}

\begin{proof}
First, we leverage Fig. \ref{figA} to prove that $\mathbf{A}_c$ is a column full rank matrix with probability one when ${\left|\mathcal{I}_c\right|}=3$ by contradiction. As shown in Fig. \ref{figA}, $O,i_0,i_1,i_2,i_3$ are the positions of UT $c$, the reference satellite and satellite $i_1,i_2,i_3$. $O_0,O_1,O_2,O_3$ are the intersection points of line $Oi_0,Oi_1,Oi_2,Oi_3$ with the unit sphere whose center is $O$, respectively. The satellites are assumed to be distributed in sphere $P_1$ following a two-dimensional distribution with a probability density function.
Since the column of $\mathbf{A}_c^T$, $\mathbf{a}_{i,c}=\frac{\mathbf{s}_c-\mathbf{s}_i}{d_{i,c}}-\frac{\mathbf{s}_c-\mathbf{s}_0}{d_{0,c}}
$, represents the difference between the unit vectors pointing to satellite $i$ and the reference satellite from UT, i.e. $\overrightarrow{OO_i}-\overrightarrow{OO_0}$, $\mathbf{a}_{1,c}$, $\mathbf{a}_{2,c}$ and $\mathbf{a}_{3,c}$ correspond to $\overrightarrow{O_0O_1}$, $\overrightarrow{O_0O_2}$ and $\overrightarrow{O_0O_3}$, respectively.
If $\mathbf{A}_c$ is not a full rank matrix, $\mathbf{a}_{3,c}$ can be represented as a linear combination of $\mathbf{a}_{1,c}$ and $\mathbf{a}_{2,c}$, which means $\overrightarrow{O_0O_3}$ should lie in  plane $P_2$ which is spanned by $\overrightarrow{O_0O_1}$ and $\overrightarrow{O_0O_2}$. As $\overrightarrow{OO_3}$ is a unit vector, $O_3$ should lie in circle $L_1$, the intersecting curve of the unit sphere and plane $P_2$.
Further, as $O_3$ is the intersection point of the line $Oi_3$ with the unit sphere, we can conclude that the point $i_3$ falls on circle $L_2$, the projection of circle $L_1$ onto sphere $P_1$ on the basis of $O$. However, the probability that the point $i_3$ falls onto circle $L_2$ is zero because the probability equals the integration of a 2-dimensional probability density function over the 1-dimensional curve. Since a contradiction occurs, the assumption that $\mathbf{A}_c$ is not a full rank matrix does not hold. Similarly, we can select 3 linearly independent columns from $\mathbf{A}_c^T$ to prove that $\text{rank}(\mathbf{A}_c)=3$ when ${\left|\mathcal{I}\right|}>3$, which means $\mathbf{A}_c$ is a column full rank matrix.
	\begin{figure}[htbp]
        \color{red}
		\centering
		\includegraphics[width=0.5\textwidth]{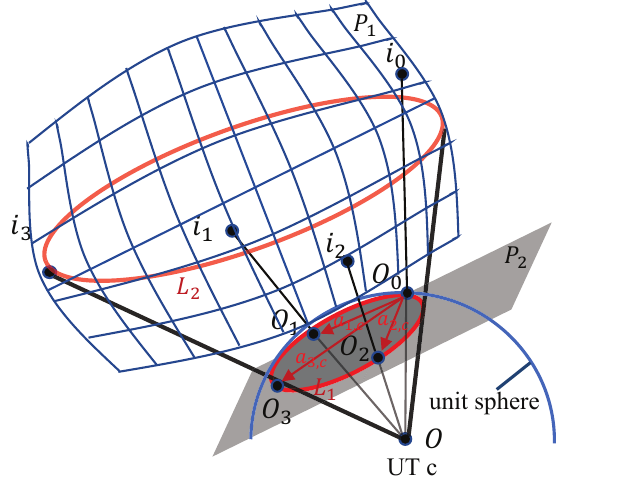}
        \centering
		\caption{Spatial relationships among UT $c$ and satellites.}
		\label{figA}
	\end{figure}
	
	Next, we show that $\mathbf{R}_c$ is positive definite. $\mathbf{R}_c$ in (\ref{4}) can be re-written as $\mathbf{R}_c=\mathbf{U}_{c}+\mathbf{V}_{c}$ as follows:
	\begin{equation*}
		\mathbf{U}_{c}=\left[
		\begin{array}{cccc}
			\sigma_{0,c}^2 & \sigma_{0,c}^2 & \cdots &\sigma_{0,c}^2 \\
			\sigma_{0,c}^2 & \sigma_{0,c}^2  & \cdots &\sigma_{0,c}^2 \\
			\vdots & \vdots & \ddots&\vdots \\
			\sigma_{0,c}^2 & \sigma_{0,c}^2 & \cdots&\sigma_{0,c}^2  \\
		\end{array}
		\right]
	\end{equation*}
 and
	\begin{equation*}
		\mathbf{V}_{c}=\left[
		\begin{array}{cccc}
			\sigma_{i_1,c}^2 &0 & \cdots &0 \\
			0 & \sigma_{i_2,c}^2  & \cdots &0\\
			\vdots & \vdots & \ddots&\vdots \\
			0 & 0& \cdots&\sigma_{i_{\left|\mathcal{I}_c\right|},c}^2 \\
		\end{array}
		\right].
	\end{equation*}
	The eigenvalues of $\mathbf{U}_{c}$ and $\mathbf{V}_{c}$ equal $\sigma_{0,c}^2,0,\cdots,0$ and $\sigma_{i_1,c}^2,\sigma_{i_2,c}^2,\cdots,\sigma_{i_{\left|\mathcal{I}_c\right|},c}^2$, respectively. Applying the Rayleigh-Ritz quotient result in \cite{17}, we can conclude that for any ${\left|\mathcal{I}\right|}$-dimensional vector $\mathbf{x}$, $0 \leq \frac{\mathbf{x}^H\mathbf{U}_{c}\mathbf{x}}{\mathbf{x}^H\mathbf{x}} \leq \sigma_{0,c}^2$ and $\min_i \sigma_{i,c}^2 \leq \frac{\mathbf{x}^H\mathbf{V}_{c}\mathbf{x}}{\mathbf{x}^H\mathbf{x}} \leq \max_i \sigma_{i,c}^2$. Thus, $$\frac{\mathbf{x}^H(\mathbf{U}_{c}+\mathbf{V}_{c})\mathbf{x}}{\mathbf{x}^H\mathbf{x}} \geq \min_i \sigma_{i,c}^2 > 0,$$ which means $\mathbf{R}_c$ is positive definite.

Further, to derive the closed-form expression of
$\frac{\partial F_c}{\partial \text{SINR}_{i,c}}$, we define $\widetilde{\sigma}_{i,c}^2=\sigma_{i,c}^2+\Delta^2$.
	Accordingly, $F_c$ and $\mathbf{R}_c$ change to $\widetilde{F}_c$ and $\widetilde{\mathbf{R}}_c$, respectively. Denote $\mathbf{v}_{\widehat{i}}=\left[0,\cdots,\Delta,\cdots,0\right]^T$ as a vector where only the $\widehat{i}$-th component equals $\Delta$ and all the others are zero. Then, using the fact that
 $$(\mathbf{X}+\mathbf{x}\mathbf{y}^T)^{-1}=\mathbf{X}^{-1}-\frac{\mathbf{X}^{-1}\mathbf{x}\mathbf{y}^T\mathbf{X}^{-1}}{1+\mathbf{y}^T\mathbf{X}^{-1}\mathbf{x}}$$
 for any vector $\mathbf{x},\mathbf{y} \in \mathbb{C}^{\vert \mathcal{I}\vert \times 1}$ and non-singular matrix $\mathbf{X}$, the expression of $\widetilde{\mathbf{R}}_c$ can be written as
	\begin{equation} \label{9}	
		\begin{split}
			\widetilde{\mathbf{R}}_c&=(\mathbf{R}_c+\mathbf{v}_{\widehat{i}}\mathbf{v}_{\widehat{i}}^T)^{-1} \\
			&=\mathbf{R}_c^{-1}-\frac{\Delta^2\mathbf{R}_c^{-1}(:,\widehat{i})\mathbf{R}_c^{-1}(\widehat{i},:)}{1+\Delta^2\mathbf{R}_c^{-1}(\widehat{i},\widehat{i})}\\
			&\xlongequal{\Delta^2 \to 0} \mathbf{R}_c^{-1}-\Delta^2\mathbf{R}_c^{-1}(:,\widehat{i})\mathbf{R}_c^{-1}(\widehat{i},:).
		\end{split}
	\end{equation}									  							
	By substituting (\ref{9}) into \eqref{4}, we have
	\begin{equation} \label{10}	
		\widetilde{F}_c=\text{tr}((\mathbf{A}_c^T\mathbf{R}_c^{-1}\mathbf{A}_c)^{-1})+\Delta^2\text{tr}( \mathbf{z}_{i,c}^T(\mathbf{A}_c^T\mathbf{R}_c^{-1}\mathbf{A}_c)^{-2}\mathbf{z}_{i,c}).
	\end{equation}	
Then,
	\begin{equation}\label{11}
		\frac{\partial F_c}{\partial \sigma_{i,c}^2}=\lim\limits_{\Delta^2\to 0} \frac{\widetilde{F}_c-F_c}{\Delta^2}=\text{tr}( \mathbf{z}_{i,c}^T(\mathbf{A}^T\mathbf{R}^{-1}\mathbf{A})^{-2}\mathbf{z}_{i,c}),
	\end{equation}
	where $\mathbf{z}_{i,c}=\mathbf{A}_c^T \mathbf{R}_c^{-1}(:,\widehat{i})$.
	By substituting (\ref{3}) into (\ref{11}) and applying the chain rule, we can get
	\begin{equation}\label{12}
		\frac{\partial F_c}{\partial \text{SINR}_{i,c}}=-\frac{3\text{tr}({\mathbf{z}_{i,c}^T(\mathbf{A}_c^T\mathbf{R}_c^{-1}\mathbf{A}_c)^{-2}\mathbf{z}_{i,c}})}{4\pi^2B^2\text{SINR}_{i,c}^2}\\.
	\end{equation}
	Since both $\mathbf{A}_c$ and $\mathbf{A}_c^T\mathbf{R}_c^{-1}(:,\widehat{i})$ are column full \text{rank} matrices,    $\mathbf{R}_c^{-1}(\widehat{i},:)\mathbf{A}_c(\mathbf{A}_c^T\mathbf{R}_c^{-1}\mathbf{A}_c)^{-2}\mathbf{A}_c^T \mathbf{R}_c^{-1}(:,\widehat{i})$ is also positive definite for positive definite matrix $\mathbf{R}_c$ \cite{18}, which means $\frac{\partial F_c}{\partial SINR_{i,c}}< 0$.
\end{proof}

Proposition 1 reveals the monotonic relationship between UT SINR of receiving the signal sent by one beam and its positioning accuracy under a specific UT-satellite-topology, and it inspires us to gradually improve $\text{SINR}_{i,c}$ to enhance overall positioning performance.
To this end, we define $\gamma_{i,c}$ as the threshold for $\text{SINR}_{i,c}$, i.e. $\gamma_{i,c}\leq \text{SINR}_{i,c}$, and successively raise $\gamma_{i,c}$ with equal step size. Specifically, we raise up $\gamma_{i,c}$ for one UT based on the derived gradient \eqref{12} as
\begin{equation}
\label{searching strategy}
\gamma_{i,c}=\gamma_{i,c}+\gamma_{0},\ c=\underset{c'\in \mathcal{C}_i} {\text{arg min}}\frac{\partial F_{c'}}{\partial \text{SINR}_{i,c'}}.
\end{equation}
For each update, the following problem is solved to check whether the new $\{\gamma_{i,c}\}_{c\in \mathcal{C}_i}$ is feasible.
\begin{equation}\label{13}
	\begin{split}
	&\mathcal{P}_2:\max\limits_{t,{\{\mathbf{w}_{i,c}\}}_{c\in \mathcal{C}_i}} t \\
	s.t.\quad &\frac{\lvert\mathbf{h}_{i,c}^T\mathbf{w}_{i,c}\rvert^2}{\sum\limits_{c'\neq c}\lvert\mathbf{h}_{i,c}^T\mathbf{w}_{i,c'}\rvert^2+\epsilon_c^2}\geq t\cdot {\gamma_{i,c}} ,\ c\in \mathcal{C}_i,\\
	& \mathbf{w}_{i,c}^H\mathbf{w}_{i,c}= P,\ c\in \mathcal{C}_i,
    \end{split}
\end{equation}
If $t^*\geq 1$, it is feasible and we can continue to raise $\gamma_{i,c}$ with \eqref{searching strategy}, while $t^*< 1$ otherwise.

Note that problem \eqref{13} can be well solved by SDR \cite{11, 19} by converting it into the following form:
\begin{equation}\label{14}
	\begin{split}
	&\mathcal{P}_3:\max\limits_{t,{\{\mathbf{Q}_{i,c}\}_{c\in \mathcal{C}_i}}} t  \\
	s.t.\quad &\frac{1}{\gamma_{i,c}}\frac{\text{tr}(\mathbf{{H'}}_{i,c}^T\mathbf{Q}_{i,c})}{\sum\limits_{c'\neq c}\text{tr}(\mathbf{{H'}}_{i,c}^T\mathbf{Q}_{i,c'})+\epsilon_c^2}\geq t,\ c\in \mathcal{C}_i,\\
	& \text{tr}(\mathbf{Q}_{i,c})= P,\ c\in \mathcal{C}_i,\\
	&\text{rank}(\mathbf{Q}_{i,c})=1,\ c\in \mathcal{C}_i,
    \end{split}
\end{equation}
where $\mathbf{{H'}}_{i,c}=\mathbf{h}_{i,c}\mathbf{h}_{i,c}^H$, $\mathbf{Q}_{i,c}=\mathbf{w}_{i,c}\mathbf{w}_{i,c}^H$, and the only non-convex constraint $\text{rank}(\mathbf{Q}_{i,c})=1$ can be relaxed. In this way, problem (\ref{14}) can be solved as a general convex optimization problem. The results of the original problem (\ref{13}) can be obtained by applying \text{rank}-one approximation based on the results of problem (\ref{14}) via $\mathbf{w}_{i,c}^*=\sqrt{\lambda^*}\mathbf{b}^*$, where $\lambda^*$ and $\mathbf{b}^*$ are the the largest eigenvalue and respective eigenvector of $\mathbf{Q}_{i,c}^*$, respectively.

Our basic description of dynamic UT SINR threshold adjustment (DSTA) based beamforming algorithm is now complete, whose detail is summarized in Algorithm 1, which assumes $T$ different values for $\gamma_{i,c}$ with an equal step size. Compared with the algorithm in \cite{19}, which has the same search space size, Algorithm 1 reduces the number of iterations from $K^T$ to $K\cdot T$. In addition, the SDR technique can be used to produce an approximated solution to \eqref{14} with a complexity of $O((N\cdot K)^{3.5})$\cite{11}. Since we need to adjust the SINR threshold $\gamma_{i,c}$ at most $O(K\cdot T)$ times with \eqref{12}, the overall time complexity of Algorithm 1 is no more than $O(N^{3.5}\cdot K^{4.5}\cdot T)$ for each satellite, which is polynomial in $N$.

\begin{algorithm}[htbp]
	\caption{DSTA-based satellite beamforming algorithm.}
	\begin{algorithmic}[1]
		\renewcommand{\algorithmicrequire}{ \textbf{Input:}}
  		\REQUIRE $\mathbf{h}_{i,c},\ \mathbf{A}_c,\ \delta_{i,c},\ \epsilon_c^2,\ \forall i  \in \mathcal{I},\ \forall c  \in \mathcal{C}$
		\renewcommand{\algorithmicrequire}{ \textbf{Output:}}
  		\REQUIRE $\mathbf{w}_{i,c},\ \forall i  \in \mathcal{I},\ \forall c  \in \mathcal{C}_i$
		\FOR{satellite $i \in \mathcal{I}$}
		\STATE Initialize $\gamma_{i,c}=\gamma'$, $c \in \mathcal{C}_i$, and all $\gamma_{i,c}$ are set to be feasible to raise.
		\FOR{UT $c\in \mathcal{C}_i$}
		\STATE Compute $\frac{\partial F_{c}}{\partial \text{SINR}_{i,c}}$ according to \eqref{12}.
		\ENDFOR
		\WHILE{there exists $\gamma_{i,c}$ which is feasible to raise and $\gamma_{i,c}<\gamma_{max}$}
        \STATE Select $c=\text{arg}\ \underset {c'}{\text{min}} \frac{\partial F_{c'}}{\partial \text{SINR}_{i,c'}}$.
        \STATE Update $\gamma_{i,c}=\gamma_{i,c}+\gamma_0$.
		\STATE Solve problem (\ref{14}) by SDR to get $t^*$ and $\mathbf{Q}_{i,c}^*$.
		\IF{$t^* \geq1$}
        \STATE Update $\text{SINR}_{i,c}$ according to \eqref{2} .
        \STATE Update $\frac{\partial F_{c}}{\partial \text{SINR}_{i,c}}$ according to \eqref{12}.
		\ELSE
		\STATE Update $\gamma_{i,c}=\gamma_{i,c}-\gamma_{0}$, and set $\gamma_{i,c}$ to be infeasible to raise.
		\ENDIF
		\ENDWHILE
		\STATE Apply \text{rank}-one approximation of  $\mathbf{Q}^*_{i,c}$ to get $\mathbf{w}^*_{i,c}$.
		\ENDFOR
	\end{algorithmic}
	\label{algorithm1}
\end{algorithm}

\section{Beam Scheduling Algorithm Design}
\label{section4}

After finishing the beamforming design under any pre-fixed beam scheduling, the next part of this paper considers the beam scheduling optimization problem, which is expressed as
\begin{equation}\label{15}
	\begin{split}
		&\mathcal{P}_4:\ \min\limits_{\{\delta_{i,c}\}} \sum\limits_{c \in \mathcal{C}} F_c(\{\delta_{i,c}\},\{\mathbf{w}^*_{i,c}(\{\delta_{i,c}\})\})  \\
		s.t.\ & C_1:\ \sum\limits_{c=1}^C\delta_{i,c}\leq K,\ i  \in \mathcal{I} ,   \\
		& C_2:\ \sum\limits_{i\in \mathcal{I}}\delta_{i,c}= I_{\text{TDOA}},\ \forall c \in \mathcal{C},       \\	
		& C_3:\ \delta_{i,c}\in \{0,1\},\ i\in \mathcal{I},\ \forall c \in \mathcal{C}.      \\	
	\end{split}
\end{equation}

 Considering the coupling of beam scheduling and beamforming design, a straightforward way to solve the above problem is to design a two-layer iterative algorithm, where the outer layer optimizes beam scheduling while the inner layer optimizes beamforming under fixed beam scheduling with our proposed DSTA-based beamforming algorithm.
 The outer layer solves a 0-1 programming, and particle swarm optimization or branch-and-bound method can be adopted. Nevertheless, such a design can incur high complexity, which is undesirable for practical implementation. Hence, in this paper, a fast and efficient beam scheduling algorithm is designed, which decouples beam scheduling from beamforming and only relies on channel state information and satellite ephemeris. To guide the design, we first examine an approximated expression of
$F_c$ to intuitively show the impacts of beam scheduling on UT positioning performance.

\begin{proposition}\label{propo2}
When $\sigma_{0,c}^2$ is small enough, $\widehat{F}_c=\text{tr}((\sum_{i\in\mathcal{I}_c}\frac{\mathbf{a}_{i,c}\mathbf{a}_{i,c}^T}{\sigma_{i,c}^2})^{-1})$ can approximate the positioning performance $F_c$ of UT $c$.
\end{proposition}

\begin{proof}
First, we show that the inequality $\mathbf{x}^T\mathbf{X}\mathbf{x}\leq\Vert \mathbf{x}\Vert_2^2\text{tr}(\mathbf{X})$ holds for any $\mathbf{x}$ and positive definite matrix $\mathbf{X}$. The eigenvalue decomposition of positive definite matrix $\mathbf{X}$ can be expressed as $\sum_{r=1}^{R}\lambda_r \boldsymbol{b}_r\boldsymbol{b}_r^T$, where $R, \lambda_r, \boldsymbol{b}_r$ are the \text{rank}, eigenvalues and eigenvectors, respectively. Applying Cauchy-Schwarz inequality, we have
\begin{equation*}
    \mathbf{x}^T\mathbf{X}\mathbf{x}=\sum_{r=1}^{R}\lambda_r \Vert\mathbf{x}^T\boldsymbol{b}_r\Vert_2^2
    \leq \sum_{r=1}^{R}\lambda_r \Vert\mathbf{x}\Vert_2^2=\Vert\mathbf{x}\Vert_2^2\text{tr}(\mathbf{X}).
\end{equation*}
		Since $\mathbf{R}_c$ can be written as $(\mathbf{V}_{c}+\sigma_{0,c}^2\mathbf{e} \mathbf{e}^T)^{-1}$, where  $\mathbf{e}=\left[1,\cdots,1\right]^T$,
  using the fact that $$(\mathbf{X}+\mathbf{x}\mathbf{y}^T)^{-1}=\mathbf{X}^{-1}-\frac{\mathbf{X}^{-1}\mathbf{x}\mathbf{y}^T\mathbf{X}^{-1}}{1+\mathbf{y}^T\mathbf{X}^{-1}\mathbf{x}},$$
  $F_c$ can be rewritten as
		\begin{equation}\label{16}
        \begin{split}
			F_c&=\text{tr}((\mathbf{A}_c^T\mathbf{R}_{c}^{-1}\mathbf{A}_c)^{-1})\\
   &=\text{tr}((\mathbf{A}_c^T\mathbf{V}_{c}^{-1}\mathbf{A}_c)^{-1})+\frac{\sigma_{0,c}^2\mathbf{u}_c^{T}(\mathbf{A}_c^T\mathbf{V}_{c}^{-1}\mathbf{A}_c)^{-2}\mathbf{u}_c}{1+\sigma_{0,c}^2\mathbf{u}_c^{T}(\mathbf{A}_c^T\mathbf{V}_{c}^{-1}\mathbf{A}_c)^{-1}\mathbf{u}_c},
        \end{split}
		\end{equation}
  where $\mathbf{u}_c=\frac{\mathbf{A}_c^T\mathbf{V}_{c}^{-1}\mathbf{e}}{\sqrt{1+\sum_{i\in \mathcal{I}_c}\sigma_{i,c}^2}}$.

  For positive definite matrices $(\mathbf{A}_c^T\mathbf{V}_{c}^{-1}\mathbf{A}_c)^{-1}$, $F_c$ in \eqref{16} can be viewed as a concave function of $\sigma_{0,c}^2$ with the derivative with respect to $\sigma_{0,c}^2$ equals $\mathbf{u}_c^{T}(\mathbf{A}_c^T\mathbf{V}_{c}^{-1}\mathbf{A}_c)^{-2}\mathbf{u}_c$ at $\sigma_{0,c}^2=0$. According to the first-order Taylor expansion of $F_c$ at $\sigma_{0,c}^2=0$, we can get
  		\begin{equation}\label{pro2}
			0\leq F_c-\widehat{F}_c\leq\sigma_{0,c}^2\mathbf{u}_c^{T}(\mathbf{A}_c^T\mathbf{V}_{c}^{-1}\mathbf{A}_c)^{-2}\mathbf{u}_c,
		\end{equation}
where $\widehat{F}_c=\text{tr}((\mathbf{A}_c^T\mathbf{V}_{c}^{-1}\mathbf{A}_c)^{-1})$. According to the conclusion $\text{tr}(\mathbf{X}^{2})\leq \text{tr}(\mathbf{X})^2$ for positive definite matrix $\mathbf{X}$ \cite{18}, we have the following inequality
\begin{equation}\label{a2}
	\begin{split}
 			\mathbf{u}_c^{T}(\mathbf{A}_c^T\mathbf{V}_{c}^{-1}\mathbf{A}_c)^{-2}\mathbf{u}_c&\leq\Vert \mathbf{u}_c\Vert_2^2\text{tr}((\mathbf{A}_c^T\mathbf{V}_{c}^{-1}\mathbf{A}_c)^{-2}) \\
  		&\leq\Vert \mathbf{u}_c\Vert_2^2\text{tr}({(\mathbf{A}_c^T\mathbf{V}_{c}^{-1}\mathbf{A}_c)^{-1}})^{2}.
    \end{split}
\end{equation}

         Since $\mathbf{a}_{i,c}$ represents the difference between two unit vectors as shown in Fig. \ref{figA}, we can get $\Vert\mathbf{a}_{i,c}\Vert_2<2$ and
\begin{equation}\label{a3}
	\begin{split}
		\Vert \mathbf{u}_c\Vert_2^2&\leq\frac{1}{1+\sum_{i\in \mathcal{I}_c}\sigma_{i,c}^2}\sum_{i\in \mathcal{I}_c}\frac{1}{\sigma_{i,c}^4}\text{tr}(\mathbf{A}_c\mathbf{A}_c^T) \\
        &\leq\sum_{i\in \mathcal{I}_c}\frac{1}{\sigma_{i,c}^4}\sum_{i\in \mathcal{I}_c}\frac{1}{v^2}\mathbf{a}_{i,c}^T\mathbf{a}_{i,c} \\
        &\leq \frac{4\mathcal{I}_{\text{TDOA}}^2}{\min_i\sigma_{i,c}^4v^2}.
    \end{split}
\end{equation}
		
By substituting \eqref{a2} and \eqref{a3} into \eqref{pro2}, we can get
\begin{equation}
\label{copy_of_pro2}
 0\leq F_c-\widehat{F}_c \leq \frac{4\mathcal{I}_{\text{TDOA}}^2\sigma_{0,c}^2}{\min_i\sigma_{i,c}^4v^2} \widehat{F}^2_c,
\end{equation}
which indicates that when $\sigma_{0,c}$ is small enough, $\widehat{F}_c=\text{tr}((\mathbf{A}_c^T\mathbf{V}_{c}^{-1}\mathbf{A}_c)^{-1})=\text{tr}((\sum_{i\in\mathcal{I}_c}\frac{\mathbf{a}_{i,c}\mathbf{a}_{i,c}^T}{\sigma_{i,c}^2})^{-1})$ can approximate $F_c$.
\end{proof}

Proposition 2 gives an approximation of UT positioning accuracy and further reveals its intrinsic relationship with beam scheduling result indicated by $\mathcal{I}_c$ that is the set of satellites serving UT $c$.
On the one hand, beam scheduling influences inter-beam/inter-user interference and hence $\sigma_{i,c}^2$ via $\text{SINR}_{i,c}$. On the other hand,
beam scheduling directly determines the relative positions among UTs and satellites, which impacts $\mathbf{a}_{i,c}$.
Inspired by these observations, a low-complexity heuristic beam scheduling algorithm is designed incorporating two steps.

In the first step, we focus on degrading inter-UT interference to achieve higher $\text{SINR}_{i,c}$, and the following metric $\rho_{i,c}$ is defined to
measure the similarity between $\mathbf{h}_{i,c}$ and the elements in $\{\mathbf{h}_{i,c'}\}_{c'\in \mathcal{C}_i}$, which is given by
\begin{equation}
	\label{17}
	\rho_{i,c}=\sum\limits_{c'\in \mathcal{C}_i\cup c}\frac{\mathbf{h}_{i,c}^T\mathbf{h}_{i,c'}}{\mathbf{h}_{i,c'}^T\mathbf{h}_{i,c'}}.
\end{equation}
Such metric is originally derived for ZF beamforming to find the users with minimum mutual interference \cite{19}.

In the second step, we consider the influence of beam scheduling on ${F}_c$ via $\mathbf{a}_{i,c}$, and the metric $\mu_{i,c}$ is defined to coarsely reflect the impact of $\mathbf{a}_{i,c}$ on ${F}_c$ as
\begin{equation}
\begin{split}
	\label{18}
	\mu_{i,c}
 &=\text{tr}((\sum_{i'\in\mathcal{I}_c\cup i}\mathbf{a}_{i',c}\mathbf{a}_{i',c}^T)^{-1})-\text{tr}((\sum_{i'\in\mathcal{I}_c}\mathbf{a}_{i',c}\mathbf{a}_{i',c}^T)^{-1})\\
 &=\sum_{r}\frac{1}{\lambda_r(\sum\limits_{i'\in\mathcal{I}_c\cup i}{\mathbf{a}_{i',c}\mathbf{a}_{i',c}^T})}-\sum_{r}\frac{1}{\lambda_r(\sum\limits_{i'\in\mathcal{I}_c}{\mathbf{a}_{i',c}\mathbf{a}_{i',c}^T})},
\end{split}
\end{equation}
where $\lambda_r(\mathbf{X})$ denotes the $r$-th eigenvalue of $\mathbf{X}$. To reduce the complexity of computing the eigenvalues in \eqref{18}, we further design a fast method for calculating $\mu_{i,c}$.

Suppose the eigenvalue decomposition of $\sum\limits_{i\in\mathcal{I}_c}{\mathbf{a}_{i,c}\mathbf{a}_{i,c}^T}$ is $  \sum_{r=1}^{R}{\lambda_r\mathbf{b}_r\mathbf{b}_r^T}$, where $R, \lambda_r, \boldsymbol{b}_r$ are the \text{rank}, eigenvalues and eigenvectors, respectively.
Moreover, $\mathbf{a}_{i,c}$ can be decomposed into two orthogonal parts. The first part is the projection onto the subspace spanned by $\{\mathbf{b}_r\}_{r=1}^R$, which can be calculated as $\sum_{r=1}^{R}\mathbf{a}_{i,c}^T\mathbf{b}_r\mathbf{b}_r$. The second part is the remaining components $\mathbf{l}_i$, which can be expressed as
\begin{equation}
	\label{19}
    \mathbf{l}_i=\mathbf{a}_{i,c} -\sum_{r=1}^{R}\mathbf{a}_{i,c}^T\mathbf{b}_r\mathbf{b}_r.
\end{equation}
Hence, we can get
\begin{equation}
    \begin{split}
    	\label{20}
    \sum\limits_{i'\in \mathcal{I}_c\cup i}{\mathbf{a}_{i',c}\mathbf{a}_{i',c}^T}&=\mathbf{a}_{i,c}\mathbf{a}_{i,c}^T+
    \sum\limits_{i'\in \mathcal{I}_c}{\mathbf{a}_{i',c}\mathbf{a}_{i',c}^T}  \\
    &=\sum_{r=1}^{R}{(\lambda_r+(\mathbf{a}_{i,c}^T\mathbf{b}_r)^2)\mathbf{b}_r\mathbf{b}_r^T}+\mathbf{l}_i\mathbf{l}_i^T.
    \end{split}
\end{equation}
Since $\mathbf{l}_i$ is orthogonal to the element in $\{\mathbf{b}_r\}_{r=1}^R$, we can get the eigenvalue decomposition expression of $\sum\limits_{i'\in \mathcal{I}_c\cup i}{\mathbf{a}_{i',c}\mathbf{a}_{i',c}^T}$ from \eqref{20}, and then when adding satellite $i$ to $\mathcal{I}_c$, the expression of rank, eigenvalue and eigenvectors of $\sum\limits_{i'\in\mathcal{I}_c}{\mathbf{a}_{i',c}\mathbf{a}_{i',c}^T}$ are updated as
\begin{subequations}
\label{21}
\begin{align}
\widetilde{R}&=\begin{cases}{R,}&{\text{if}}\ \Vert\mathbf{l}_i \Vert_2=\mathbf{0},\\
R+1,&\Vert\mathbf{l}_i \Vert_2\neq \mathbf{0},\\
\end{cases}\\
    \widetilde{\lambda}_r&=\begin{cases}\lambda_r+(\mathbf{a}_{i,c}\mathbf{b}_r)^2,&\ r=1,\cdots,R, \\
{\Vert\mathbf{l}_i \Vert_2^2},&\ r=R+1,\ {\text{if}}\ \Vert\mathbf{l}_i \Vert_2\neq \mathbf{0},
\end{cases} \\
\widetilde{\mathbf{b}}_{r}&=\begin{cases}{\mathbf{b}_r},&\ r=1,\cdots,R,\\
\mathbf{l}_i /\Vert\mathbf{l}_i
\Vert_2,&\ r=R+1,\ {\text{if}}\ \Vert\mathbf{l}_i \Vert_2\neq \mathbf{0}.\\
\end{cases}
\end{align}
\end{subequations}

With the above easily calculated equations, \eqref{18} can be fast computed with
\begin{equation}
	\label{22}
	\mu_{i,c}=
	\sum\limits_{r=1}^{\widetilde{R}}\frac{1}{\widetilde{\lambda}_r}-\sum_{r=1}^{R}\limits\frac{1}{\lambda_r}.
\end{equation}

Based on the above metric, instead of finding the optimal beam scheduling iteratively, we only need to find a beam scheduling plan that satisfies the constraints of \eqref{15} and provides good channel orthogonality and UT-satellite-topology geometry. Specifically, we first greedily select $m$ candidate satellites for each UT in each iteration according to the principle of minimizing channel similarity $\rho_{i,c}$. Then, with the principle of maximizing $\mu_{i,c}$, a satellite is finally selected to provide service for UT $c$. Finally, $\mathcal{I}_c$ and $\mathcal{C}_i$ are updated. The summary of the proposed heuristic beam scheduling (HBS) algorithm is given in Algorithm 2, whose complexity is $O(C\cdot I_{\text{TDOA}}\cdot I\cdot K)$.

Overall, we design a low-complexity joint beamforming and beam scheduling algorithm with the help of Proposition 1 and Proposition 2, which essentially reveal the correlation of positioning performance with SINR and GDOP. Compared with the genetic algorithms \cite{res6}, HBS-DSTA need not feed the output of the beamformer back to the scheduler with the complexity of $O(I\cdot N^{3.5}\cdot K^{4.5}\cdot T)+O(C\cdot I_{\text{TDOA}}\cdot I\cdot K)$.

\begin{algorithm}[htbp]
	\caption{The HBS algorithm.}
	\begin{algorithmic}[1]
		\renewcommand{\algorithmicrequire}{ \textbf{Input:}}
		\REQUIRE $\mathbf{h}_{i,c}, \mathbf{s}_i, \mathbf{s}_c, \mathbf{s}_0, \forall i  \in \mathcal{I},  \forall c \in \mathcal{C}$
		\renewcommand{\algorithmicrequire}{ \textbf{Output:}}
		\REQUIRE $\delta_{i,c}, \forall i  \in \mathcal{I},\forall c  \in \mathcal{C}$
		\STATE Initialize $\delta_{i,c}=0,\mathcal{C}_{i}=\emptyset, \mathcal{I}_{c}=\emptyset, K_i=0 ,\forall i \in \mathcal{I},\forall c \in \mathcal{C}$.
		\FOR{UT $c\in \mathcal{C}$}
		\WHILE{$\vert\mathcal{I}_c\vert< I_{\text{TDOA}}$}
		\FOR{satellite $i \in \mathcal{I}$}
		\IF{$K_i<K$}
		\STATE
		Compute similarity $\rho_{i,c}$ according to \eqref{17}.
		\STATE Compute $\widetilde{R},
		\widetilde{\lambda}_{r}, \widetilde{\mathbf{b}}_{r}$ according to \eqref{21}.
        \STATE Compute $\mu_{i,c}$ with the fast calculation method  according to \eqref{22}.
		\ENDIF
		\ENDFOR
		\STATE Select the set $\mathcal{I}^*$ of $m$ candidate satellites with smallest $\rho_{i,c}$. 
		\STATE Select the satellite $i^*$ to form the best UT-satellite topology geometry with $i^*=\text{arg}\ \underset {i\in \mathcal{I}^*}{\text{max}}\ \mu_{i,c}$.
  		\STATE Compute $\widetilde{R},
		\widetilde{\lambda}_{r}, \widetilde{\mathbf{b}}_{r}$ according to \eqref{21}.
        \STATE Update $R,\lambda_r,\mathbf{b}_r$ with $\widetilde{R},\widetilde{\lambda}_r,\widetilde{\mathbf{b}}_r$.
		\STATE Allocate a beam of satellite $i^*$ to UT $c$. Set $\delta_{i^*,c}=1, \mathcal{C}_{i^*}=\mathcal{C}_{i^*}\cup \{c\}, \mathcal{I}_c=\mathcal{I}_c\cup \{i^*\}, K_{i^*}=K_{i^*}+1$.
		\ENDWHILE
		\ENDFOR
	\end{algorithmic}
	\label{algorithm 2}
\end{algorithm}

\section{SIMULATION RESULTS}
\label{sec_SIMULATION_RESULTS}

This section presents simulation results to verify the performance of the proposed DSTA-based beamforming and beam scheduling algorithms. We choose three beamforming algorithms for comparison as below.
\begin{itemize}
\item ZF based beamforming \cite{21}. ZF beamforming eliminates the co-channel interference (CCI) term in \eqref{2} caused by frequency reuse among beams. Denoting the matrix of channel by $\mathbf{H}_i=[\mathbf{h}_{i,c_1},\cdots,\mathbf{h}_{i,c_{\vert\mathcal{C}_i\vert}}]^H\in \mathbb{C}^{\vert\mathcal{C}_i\vert\times N}$, the beamforming weight matrix $\mathbf{W}_i$ of satellite $i$ is calculated by
\begin{equation}
\label{ZF}
\mathbf{W}_i=\beta \mathbf{H}_i^H(\mathbf{H}_i\mathbf{H}_i^H)^{-1}\\
\end{equation}
with
\begin{equation*}
\beta=\sqrt{\frac{P\vert\mathcal{C}_i\vert}{\Vert(\mathbf{H}_i^H\mathbf{H}_i)^{-1}\Vert_F^2}}.
\end{equation*}
\item Single-cell beamforming (SCB) \cite{17}.
In single-cell beamforming, $\mathbf{w}_{i,c}$ is calculated by
\begin{equation}
\label{SCB}
\mathbf{w}_{i,c}=\sqrt{\frac{P}{\Vert \mathbf{h}_{i,c}\Vert_F^2}}\mathbf{h}_{i,c}.
\end{equation}
\item Single-cell beamforming without interference (SCBWI). For comparison, we consider an ideal scenario where inter-beam interference is completely ignored and $\mathbf{w}_{i,c}$ is calculated the same as SCB. Since each UT SINR reaches the theoretical maximum, individual positioning accuracy is also the best, according to Proposition 1. Therefore, we use it as the ideal upper bound of the positioning accuracy. At this time, $\text{SINR}_{i,c}$ is reduced to $\text{SNR}_{i,c}$ that is calculated by
\begin{equation}
\label{NI}
	\text{SNR}_{i,c}=\frac{\lvert\mathbf{h}_{i,c}^T\mathbf{w}_{i,c}\rvert^2}{\epsilon_c^2}.
\end{equation}
\end{itemize}

In the baseline schemes for beam scheduling, we have compared the proposed HBS algorithm with the GDOP-based scheduling scheme \cite{22}, the communication-oriented beam scheduling scheme \cite{19}, and the parallax-based scheduling scheme \cite{res4}.
\begin{itemize}
\item GDOP-based scheduling. GDOP-based beam scheduling has been widely used in Global Positioning System. In GDOP-based scheduling, satellites are continuously and greedily selected from $\mathcal{I}$ to minimize GDOP for UT $c$ until there are $I_{\text{TDOA}}$ satellites providing service for UT $c$. GDOP is calculated as
\begin{equation}
\label{23}
	\text{GDOP}_c=\frac{1}{v}\sqrt{\text{tr}(({\mathbf{A}_c^T\mathbf{A}_c})^{-1})}.
\end{equation}	
Note that $\text{GDOP}_c$ can also be calculated by converting it into the form with regard to the eigenvalues of $\mathbf{A}_c$. With the proposed method \eqref{21}, $\text{GDOP}_c$ can be fast computed. Meanwhile, the GDOP-based beam scheduling is a particular case of our proposed beam scheduling algorithm if we set
$m=\vert\mathcal{I}\vert$.

\item Communication-oriented scheduling. In communication-oriented scheduling, satellites are continuously and greedily selected from $\mathcal{I}$ to minimize $\rho_{i,c}$ until there are $I_{\text{TDOA}}$ satellites providing service for UT $c$. Meanwhile, communication-oriented scheduling is a special case of HBS if we set $m=1$.

\item Parallax-based scheduling. We first denote
\begin{equation}
\mathbf{P}_c=	\left[
		\begin{array}{cccc}
			cos\phi_{1,1}^c & cos\phi_{1,2}^c & \cdots &cos\phi_{1,I}^c \\
			cos\phi_{2,1}^c & cos\phi_{2,2}^c  & \cdots &cos\phi_{2,I}^c \\
			\vdots & \vdots & \ddots&\vdots \\
			cos\phi_{I,1}^c & cos\phi_{I,2}^c & \cdots& cos\phi_{I,I}^c  \\
		\end{array}
		\right],
\end{equation}
and cost function $\mathbf{J}_i=\sum_{j=1}^{I}cos\phi_{i,j}^c$, where $\phi_{i,j}^c$ is the parallax of UT $c$ observed from satellite $i$ and satellite $j$. In parallax-based scheduling, we continuously delete satellite $i$ corresponding to max $\mathbf{J}_i$ and remove the row and column elements of the matrix $\mathbf{P}_c$ corresponding to the removed satellite until there are $I_{\text{TDOA}}$ satellites remaining.
\end{itemize}

Key simulation parameters are summarized in Table \ref{tab1}, where we set the number of satellites $I$ visible to the ground user to $16, 21, 26$ with the corresponding number of satellites/beams serving each UT being $3, 4, 5$.

\begin{table}[htbp]
	\centering
	\caption{Key simulation parameters }
	\begin{tabular}{cc}\hline
		Parameters & Values \\\hline
		The number of satellites visible to UTs & $16, 21 ,26$\\
  		$I_{\text{TDOA}}$ &$3, 4 ,5$\\
        $\sigma_{0,c}^2$     & $10^{-19}$                    \\
		The number of cells & $61$ \\
        Cell radius & $43.3$ km \\
		The maximum number of beams of each satellite & $12$\\
		The number of antennas per satellite& $8\times8$\\
		Orbit height& $600$ km\\
		Maximum beam transmission power& $20, 23, 26$ dBw\\
		Carrier frequency	&$4$ GHz\\
		Bandwidth for each satellite &$50$ Mhz\\
        UT's receiving noise power density &$-174$ dBm/Hz\\
		Receiving antenna gain of UTs &  $0$ dBi\\
        \hline
	\end{tabular}
	\label{tab1}
\end{table}

\subsection{The Effectiveness of DSTA-based Beamforming}
\label{Simulation-results}

When $I_{\text{TDOA}}=4$, the cumulative distribution functions (CDFs) of UT positioning accuracy under different beamforming schemes are shown in Fig. \ref{fig2}, where the proposed HBS is adopted with $m=4$ for all the schemes.
It can be found that DSTA-based beamforming achieves higher positioning accuracy than the other schemes employing positioning-oriented beamforming.
As shown in Table \ref{tab2}, the average positioning accuracy of DSTA-based beamforming is around $9.54$m in terms of CRLB. It significantly outperforms ZF and SCB by $11.9\%$ and $51.1\%$, respectively. In addition, the performance of our proposed scheme is just within a small gap from the version of SCBWI.

\begin{figure}[htbp]
	\centering
	\includegraphics[width=0.5\textwidth]{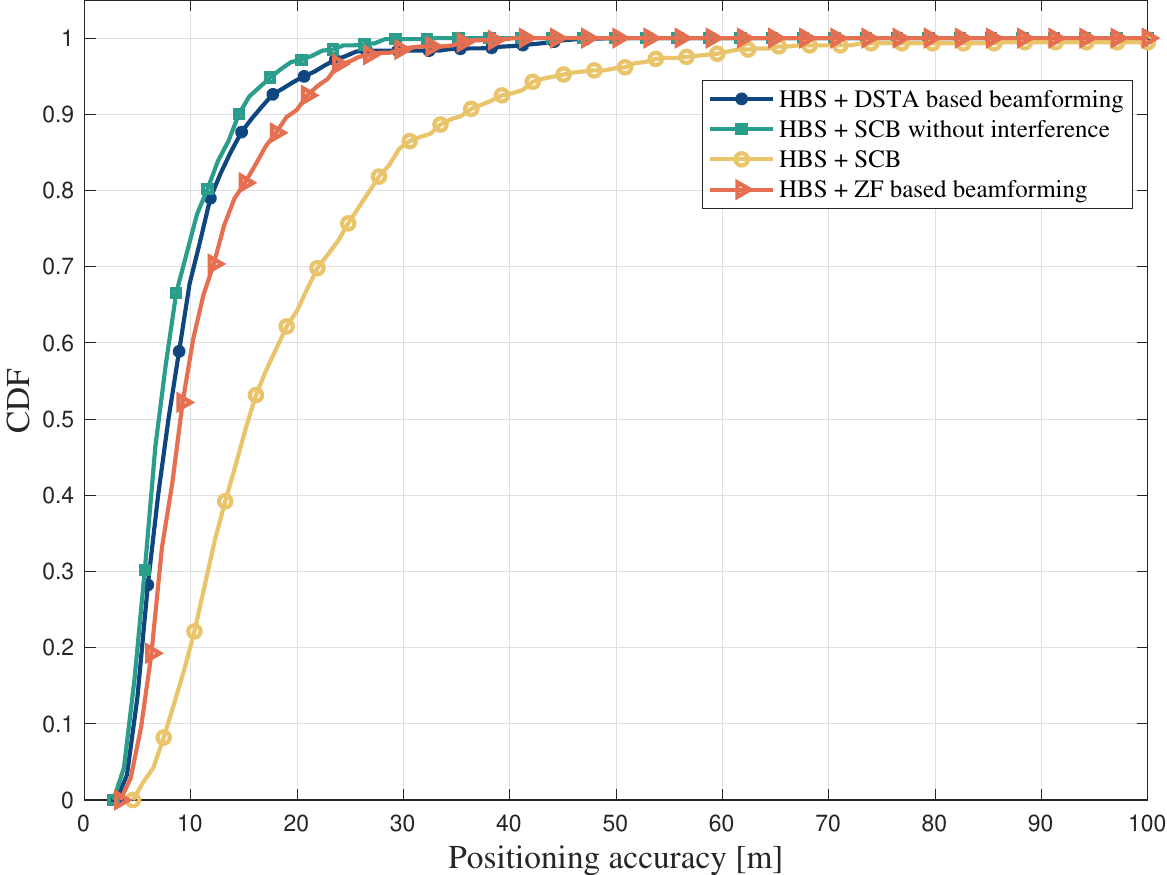}
	\caption{The CDFs of UT positioning accuracy under different beamforming schemes when $I=21$, $I_{\text{TDOA}}=4$, and $P=26$ dBw.}
	\label{fig2}
\end{figure}

\begin{table}[htbp]
	\centering
	\caption{Average UT positioning accuracy with DSTA, SCBWI, ZF, and SCB beamforming schemes under different system settings (m).}
	\begin{tabular}{ccccc}\hline
		Schemes & DSTA & SCBWI & ZF &SCB\\\hline
        P=20 dBw, $I_{\text{TDOA}}=4$ &16.5090& 16.5626& 19.9165&24.0789\\
		P=23 dBw, $I_{\text{TDOA}}=4$ &13.2239& 11.8581&14.2049&20.7337\\
        P=26 dBw, $I_{\text{TDOA}}=3$&11.4812&\phantom{0}9.8182&13.5972&29.9317\\
		P=26 dBw, $I_{\text{TDOA}}=4$&\phantom{0}9.5415&\phantom{0}8.4220&10.8467&19.5422\\
		P=26 dBw, $I_{\text{TDOA}}=5$& \phantom{0}8.0563&\phantom{0}7.2250&\phantom{0}8.9308&15.7821\\\hline
	\end{tabular}
	\label{tab2}
\end{table}

The CDFs of UT SINR under four beamforming schemes are drawn in Fig. \ref{fig3} to give more insight into the superior performance of our proposal.
Notably, it indicates a significant improvement brought by DSTA in SINR distribution compared to ZF and SCB schemes,
and there is only a 1-3dB performance loss compared to SCBWI.
Meanwhile, the stepped curve with DSTA-based beamforming is because we raise the SINR threshold by fixed value $\gamma_0$ each time.

\begin{figure}[htbp]
	\centering
	\includegraphics[width=0.5\textwidth]{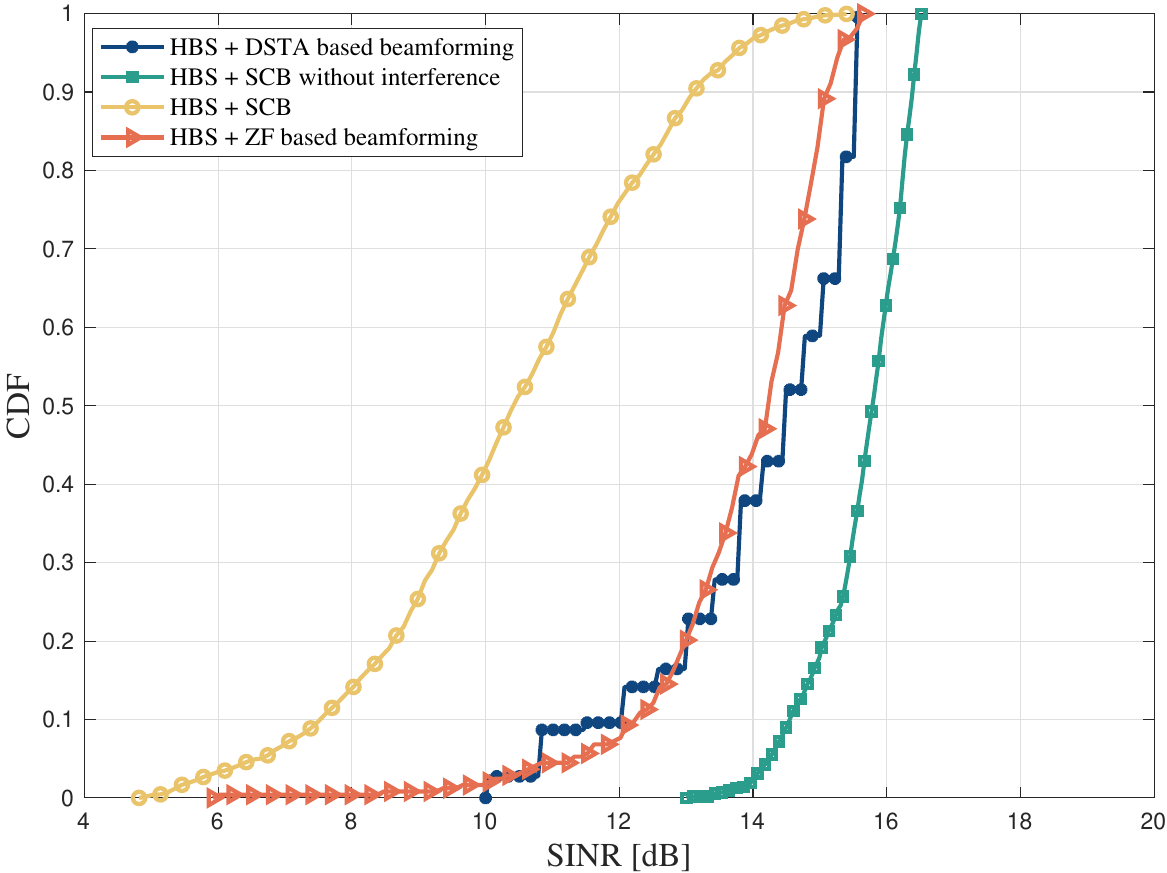}
	\caption{The CDFs of UT SINR under different beamforming schemes when $I=21$, $I_{\text{TDOA}}=4$, and $P=26$ dBw.}
	\label{fig3}
\end{figure}

\subsection{The Impact of Beam Transmission Power Budget on The Performance of DSTA-based Beamforming}

In Fig. \ref{fig5}, the performance of the DSTA-based beamforming scheme is evaluated with $I_{\text{TDOA}}=4$ under various satellite transmission power budgets. It can be found that the positioning accuracy performance degrades when the available transmission power decreases, which is possible because UT SINR becomes lower at this time. Moreover, according to Table \ref{tab2}, when the transmission power budget decreases from 26 dBw to 20 dBw, DSTA outperforms ZF more significantly, which corresponds to the positioning accuracy improvement by $11.9\%$ to $17.1\%$. The reason for such a performance difference is that interference power is not the key that influences UT SINR and  UT positioning accuracy if beam transmission power gets lower.

\begin{figure}[htbp]
	\centering
	\includegraphics[width=0.5\textwidth]{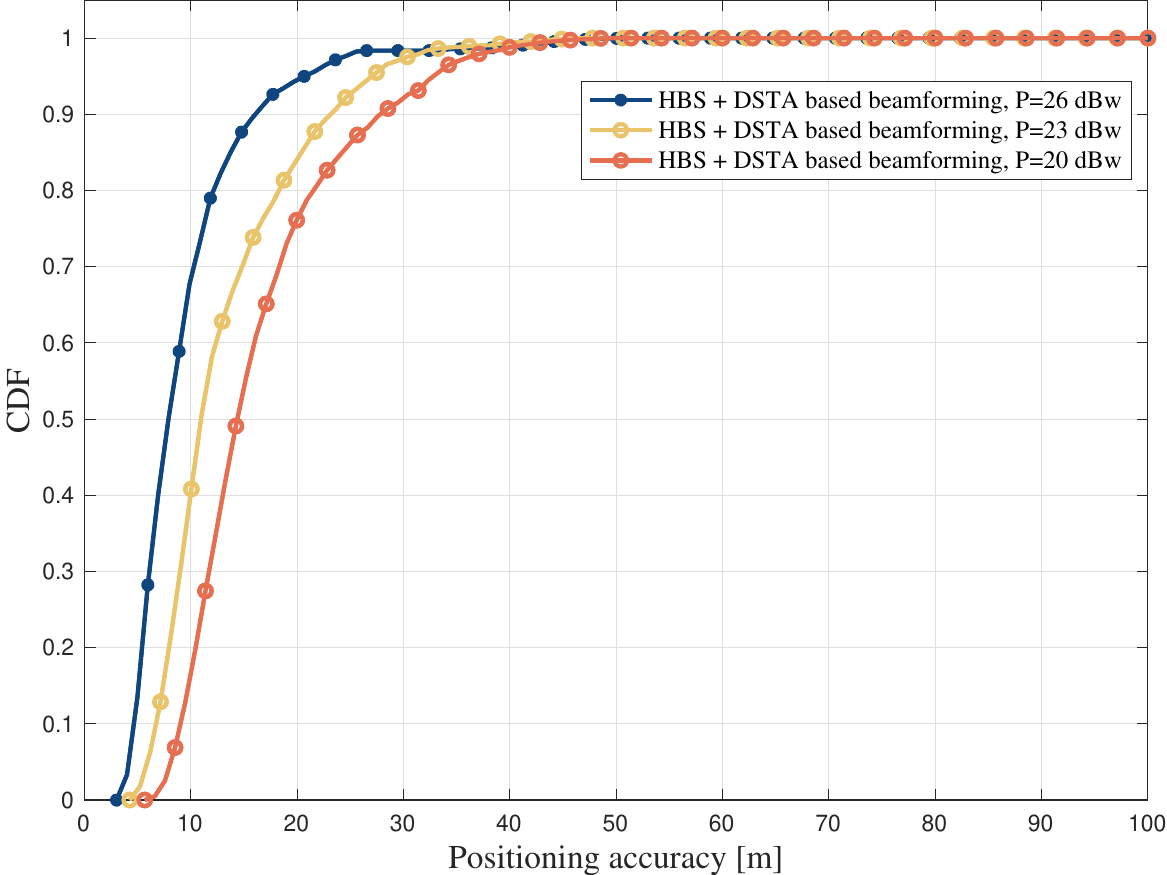}
	\caption{The CDFs of positioning accuracy with DSTA-based beamforming under different satellite transmission power budgets when $I=21$ and $I_{\text{TDOA}}=4$.}
	\label{fig5}
\end{figure}

\subsection{The Impact of The Number of Visible Satellites on The Performance of DSTA-based Beamforming}

 The CDFs of UT positioning accuracy when $I_{\text{TDOA}}= 3, 4, 5$ ($I=16,21,26$) are shown in Fig. \ref{fig6}, and it can be found that a more significant number of visible satellites can contribute to a considerable improvement in positioning performance. As shown in Table \ref{tab2}, when the number of satellites serving the UT increases from 3 to 5, the positioning accuracy with DSTA improves by $29.8\%$. This phenomenon is expected because more good-quality signals can be measured, and better network topology geometry can be obtained. To verify the performance gain of our proposal, we list the accuracy performance under different beamforming schemes and $I_{\text{TDOA}}$ in Table \ref{tab2}.

\begin{figure}[htbp]
	\centering
	\includegraphics[width=0.5\textwidth]{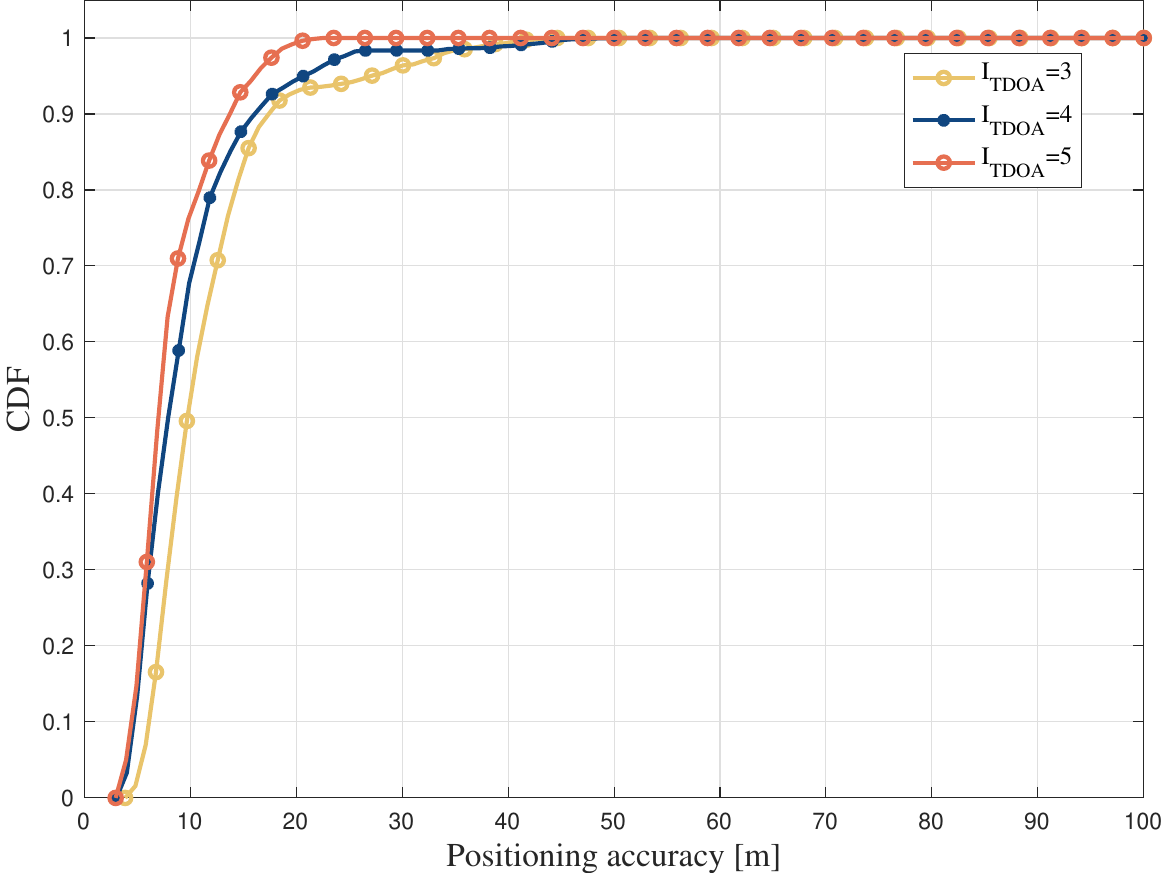}
	\caption{The CDFs of positioning accuracy with DSTA-based beamforming under different numbers of visible satellites when $P=26$ dBw.}
	\label{fig6}
\end{figure}

\subsection{Performance Evaluation of The Proposed Heuristic Beam Scheduling}

To examine the performance of the proposed beam scheduling scheme, Fig. \ref{fig7} shows the CDFs of UT positioning accuracy under GDOP-based beam scheduling and our proposal with $m=1,4,12$ and it can be found that the proposed scheme outperforms GDOP-based beam scheduling, which is more significant when the parameter $m$ is appropriately selected. Moreover, as shown in Table \ref{tab3}, the proposed HBS scheme always results in substantial gains in positioning performance with different beamforming schemes. Specifically, it outperforms GDOP-based beam scheduling by $28.3\%$, $55.9\%$, and $53.2\%$ in average positioning accuracy with DSTA-based beamforming for $m=1,4,12$, respectively.
The above results are because the GDOP-based scheme only considers the UT-satellite topology geometry. In contrast, the HBS scheme considers the impact of geometric topology property and channel correlation among UTs on positioning performance.

\begin{table}[htbp]
	\centering
	\caption{Average UT positioning accuracy under different beam scheduling schemes when $I=21$, $I_\text{TDOA}=4$, and $P=26$ dBw. (m)}
	\begin{tabular}{ccccc}\hline
		Schemes & DSTA & SCBWI & ZF &SCB\\\hline
        \phantom{00}HBS, $m=1$&15.1943&13.9944&15.2543&26.4950\\
		\phantom{00}HBS, $m=4$&\phantom{0}9.5415&\phantom{0}8.4220&10.8467&19.5422\\
        \phantom{000}HBS, $m=12$&10.1288&\phantom{0}6.3246&12.8647&26.0309\\
		\phantom{000}GDOP-based& 21.6105&\phantom{0}6.0365&91.5679&47.9319\\
  	    \phantom{000}Parallax-based& 34.0121&15.2163&181.157&103.7096\\\hline
	\end{tabular}
	\label{tab3}
\end{table}

At last, to further reveal the intrinsic influence of parameter $m$ in the proposed HBS scheme, Fig. \ref{fig8} and Fig. \ref{fig9} show the CDFs of UT SINR and GDOP under different $m$. In Fig. \ref{fig8}, it is demonstrated that smaller $m$ contributes to higher UT SINR, and HBS-DSTA achieves the best UT SINR performance with $m=1$.
However, this does not necessarily result in the positioning performance. As $m$ becomes smaller, as shown in Fig. \ref{fig9}, the GDOP value increases rapidly due to the formation of improper UT-satellite topology. Although the communication-oriented scheme exhibits better SINR and GDOP-based and parallax-based schemes exhibit better GDOP, HBS has superior positioning performance. These results reflect there exists a trade-off between UT SINR and GDOP value that can be adjusted by $m$.

\begin{figure}[htbp]
	\centering
	\includegraphics[width=0.5\textwidth]{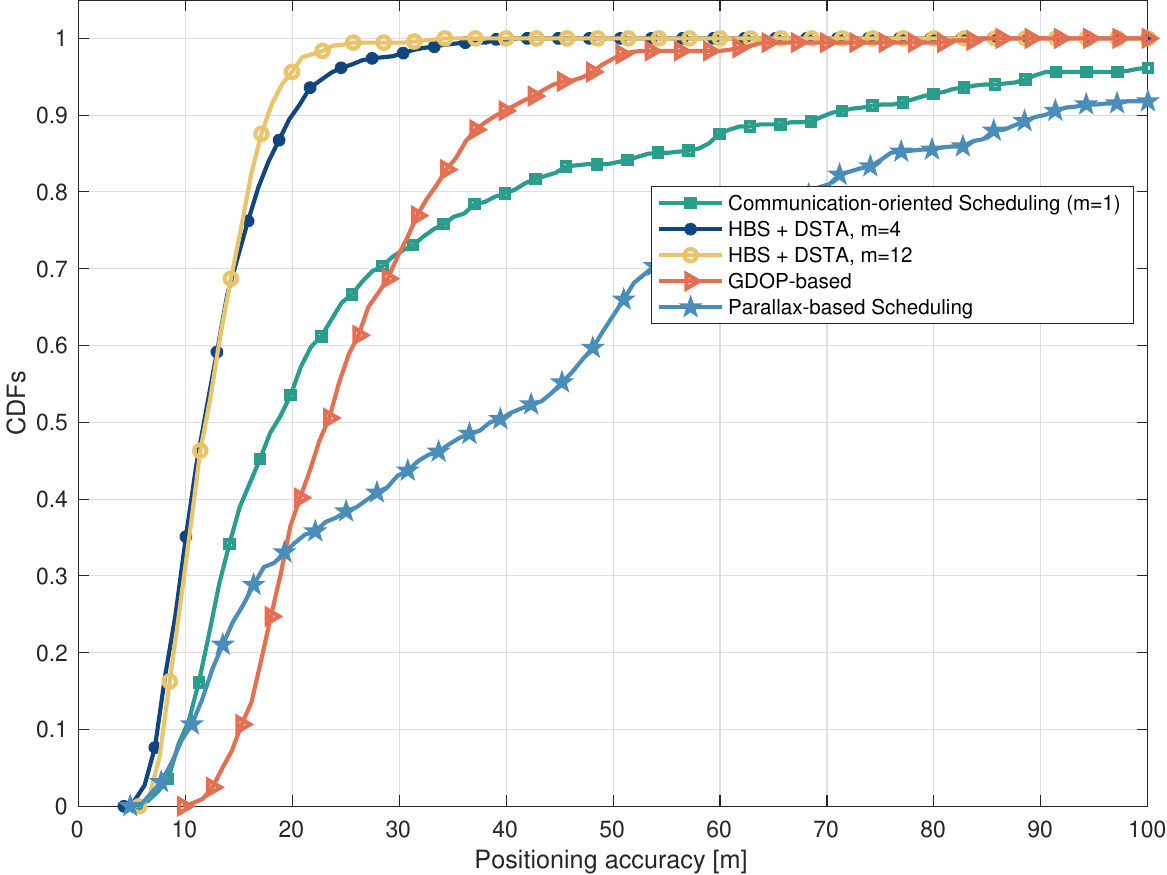}
	\caption{The CDFs of UT positioning accuracy under different beam scheduling schemes when $I=21$, $I_{\text{TDOA}}=4$, and $P=26$ dBw.}
	\label{fig7}
\end{figure}

\begin{figure}[htbp]
	\centering
	\includegraphics[width=0.5\textwidth]{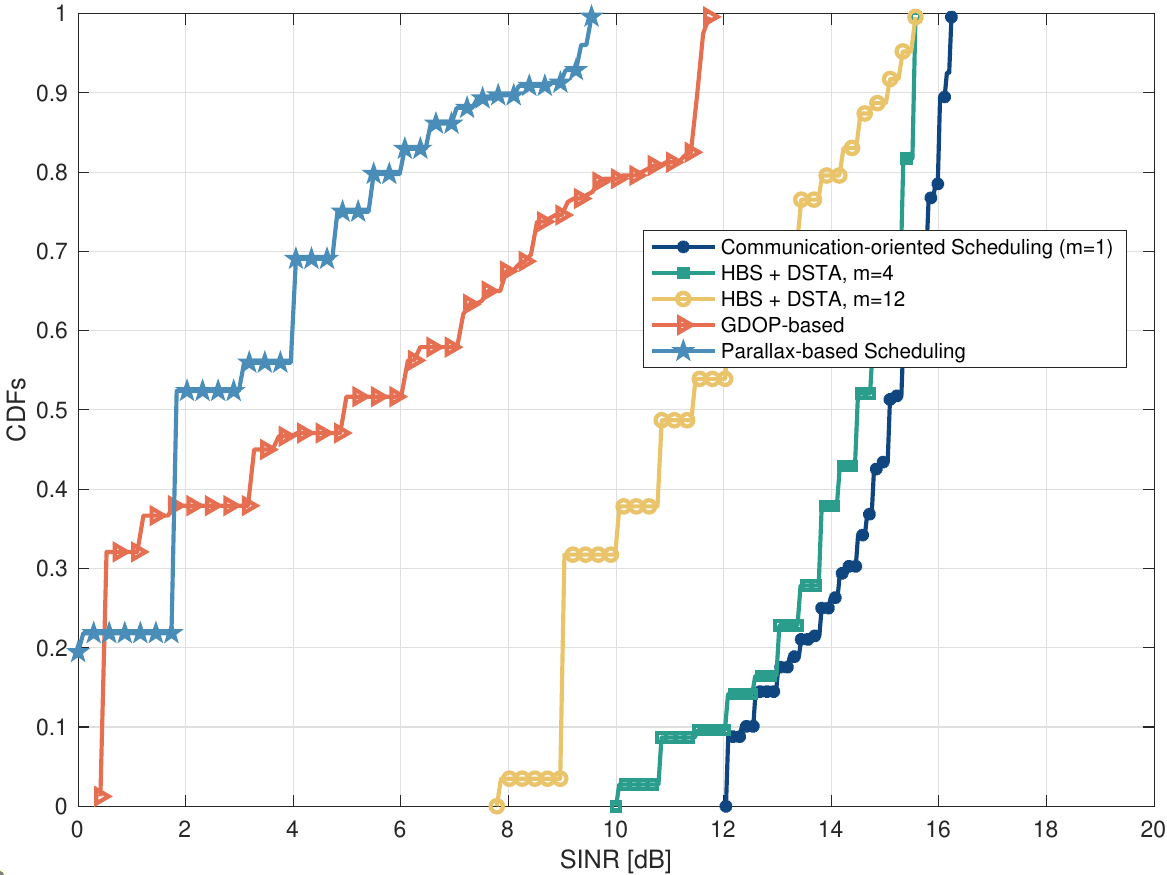}
	\caption{The CDFs of SINR under different beam scheduling schemes when $I=21$, $I_{\text{TDOA}}=4$, and $P=26$ dBw.}
	\label{fig8}
\end{figure}

\begin{figure}[htbp]
	\centering
	\includegraphics[width=0.5\textwidth]{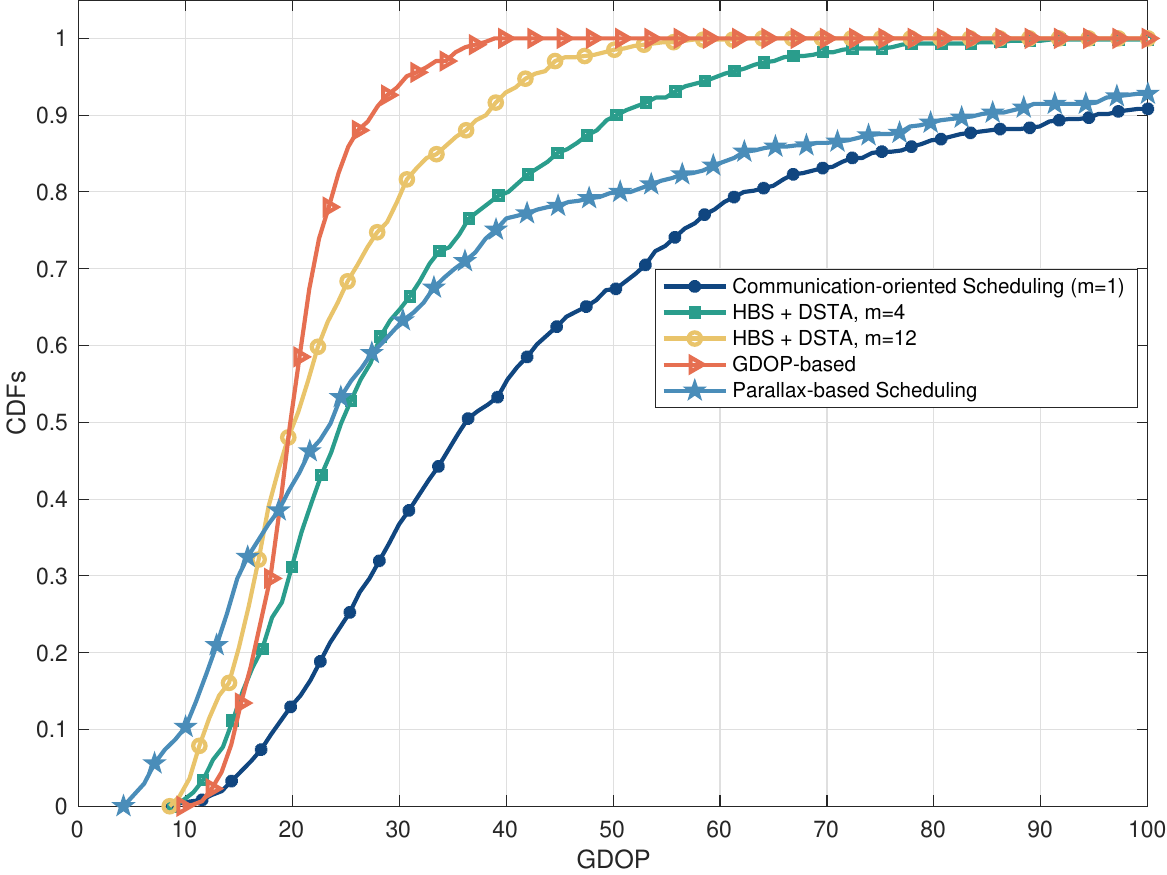}
	\caption{The CDFs of UT GDOP under different beam scheduling schemes when $I=21$, $I_{\text{TDOA}}=4$, and $P=26$ dBw.}
	\label{fig9}
\end{figure}

\section{CONCLUSION}
\label{section5}

 In this paper, positioning-oriented beam scheduling and beamforming have been studied in multi-beam LEO satellite networks. As user positioning performance suffers from severe inter-beam interference and improper UT-satellite topology geometry, based on the monotonic relationship between user positioning accuracy and perceived SINR, we have designed a dynamic SINR adjustment based beamforming algorithm with SDR technique and also a fast heuristic beam scheduling algorithm that takes channel correlation among users and satellite-user topology geometric property into account. By simulation, it has been shown that average user positioning accuracy can be improved by
17.1\% and 55.9\%, compared to conventional beamforming and beam scheduling schemes, respectively.


\begin{IEEEbiography}[{\includegraphics[width=1in,height=1.25in,clip,keepaspectratio]{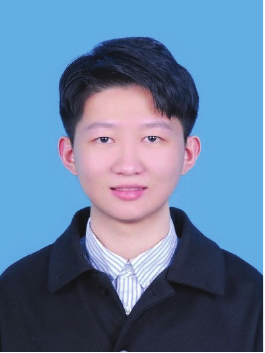}}]{Hongtao Xv} 
received the B.S. degree in information and communication engineering from Beijing University of Posts and Telecommunications (BUPT), Beijing, China, in 2022. He is currently working toward a Ph.D. degree in the State Key Laboratory of Networking and Switching Technology, BUPT, Beijing, China. His research interests include LEO satellite communication and beamforming. 
\end{IEEEbiography}

\vspace{-20 mm} 
\begin{IEEEbiography}
[{\includegraphics[width=1in,height=1.25in,clip,keepaspectratio]{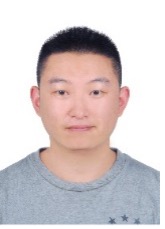}}]{Yaohua Sun} received the bachelor’s degree (Hons.) in telecommunications engineering (with management) and the Ph.D. degree in communication engineering from Beijing University of Posts and Telecommunications (BUPT), Beijing, China, in 2014 and 2019, respectively. He is currently an Associate Professor with the School of Information and Communication Engineering, BUPT. His research interests include intelligent radio access networks and LEO satellite communication. He has published over 30 papers, including 3 ESI highly cited papers. He has been a Reviewer for \textit{IEEE Trans. Commun.}, \textit{IEEE Trans. Mob. Comput.}, etc. 
\end{IEEEbiography}

\vspace{-20 mm} 
\begin{IEEEbiography}[{\includegraphics[width=1in,height=1.25in,clip,keepaspectratio]{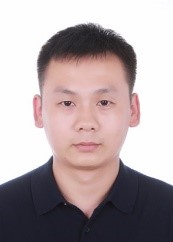}}]{Yafei Zhao} 
received his B.S., M.S., and Ph.D. degrees in aeronautical and astronautical science and technology from the Harbin Institute of Technology (HIT), China, in 2010, 2012, and 2018, respectively. He is currently an associate professor at Beijing University of Posts and Telecommunications. His research interests include integrated communication and navigation, integrated satellite-terrestrial networks, etc.
\end{IEEEbiography}

\vspace{-20 mm} 
\begin{IEEEbiography}[{\includegraphics[width=1in,height=1.25in,clip,keepaspectratio]{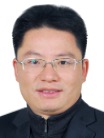}}]{Mugen Peng} (Fellow, IEEE) received the Ph.D. degree in communication and information systems from the Beijing University of Posts and Telecommunications, Beijing, China, in 2005. In 2014, he was an Academic Visiting Fellow at Princeton University, Princeton, NJ, USA. He joined BUPT, where he has been the Dean of the School of Information and Communication Engineering since June 2020, and the Deputy Director of the State Key Laboratory of Networking and Switching Technology since October 2018. He leads a Research Group focusing on wireless transmission and networking technologies with the State Key Laboratory of Networking and Switching Technology, BUPT. His main research interests include wireless communication theory, radio signal processing, cooperative communication, self-organization networking, integrated sensing and communication, and Internet of Things. He was a recipient of the 2018 Heinrich Hertz Prize Paper Award, the 2014 IEEE ComSoc AP Outstanding Young Researcher Award, and the Best Paper Award in IEEE ICC 2022, JCN 2016, and IEEE WCNC 2015. He is on the Editorial or Associate Editorial Board of \textit{IEEE Commun. Mag.}, \textit{IEEE Netw.}, \textit{IEEE Internet Things J.}, \textit{IEEE Trans. Veh. Technol.}, and \textit{IEEE Trans. Netw. Sci. Eng.}, etc.
\end{IEEEbiography}

\vspace{-20 mm} 
\begin{IEEEbiography}[{\includegraphics[width=1in,height=1.25in,clip,keepaspectratio]{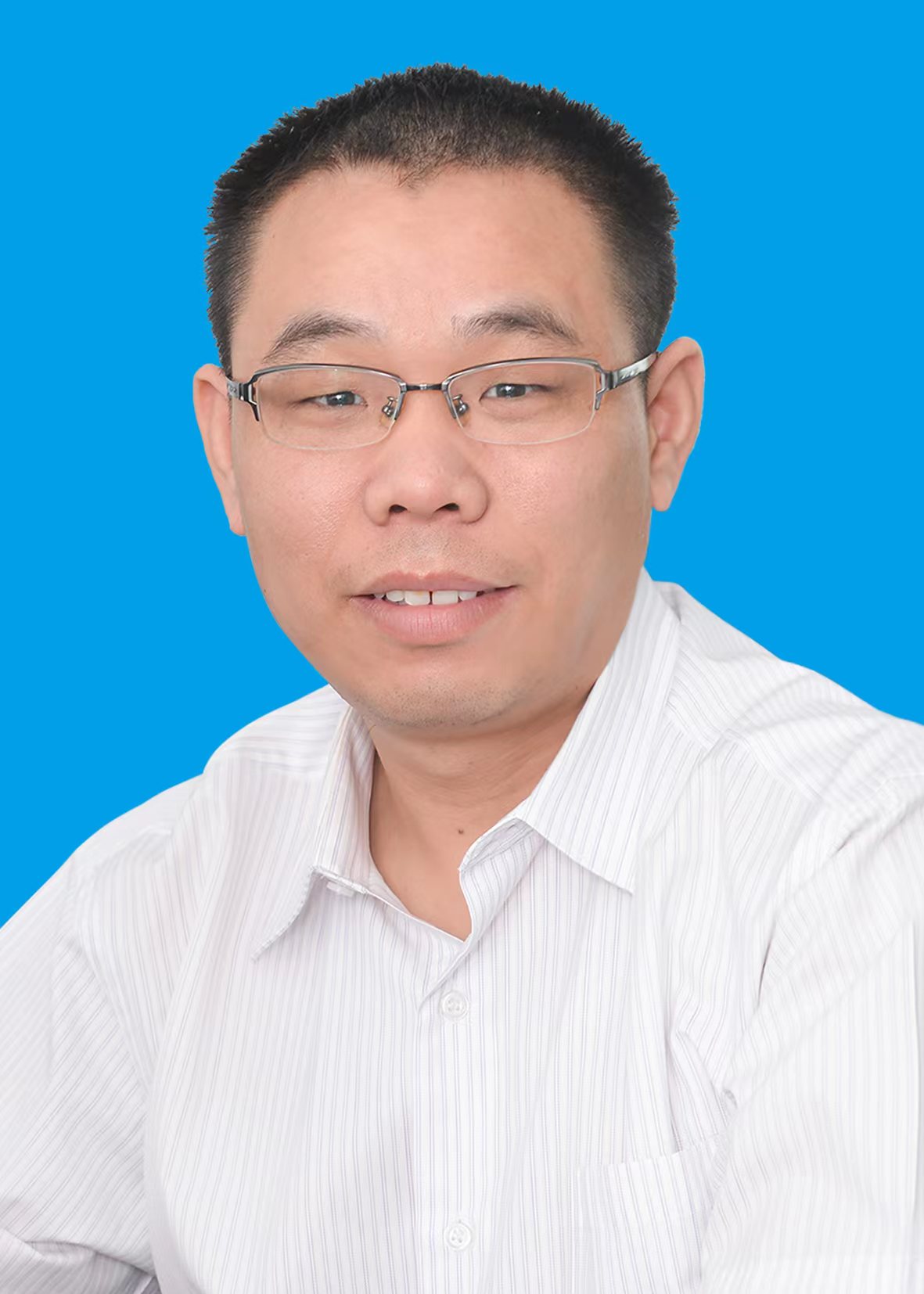}}]{Shijie Zhang} received the B.S., M.S., and Ph.D. degrees in spacecraft design from the Harbin Institute of Technology, Harbin, China, in 2000, 2002, and 2005, respectively. He is currently the chief scientist of the Yinhe Hangtian (Beijing) Communication Technology Co., Ltd., and the professor of the State Key Laboratory of Media Convergence Production Technology and Systems. His research interests include satellite internet, aerospace communication, spacecraft formation flying, etc.
\end{IEEEbiography}

\end{document}